\documentclass[svgnames,11pt,a4paper]{article}

\usepackage{amssymb,amsmath,amsthm}
\usepackage{vmargin}

\setmarginsrb{1.1in}{1.1in}{1.1in}{1.1in}{0pt}{0pt}{0pt}{7mm}

\usepackage{amstext}

\newtheorem{brule}{B}	
\newtheorem{theorem}{Theorem}

\newtheorem{corollary}{Corollary}
\newtheorem{lemma}{Lemma}
\newtheorem{claim}{Claim}

\newtheorem{define}{Definition}

\usepackage{xcolor}
\usepackage{epsfig}
\usepackage{boxedminipage}
\usepackage{boxedminipage}
\usepackage{xcolor}
\usepackage{framed}

\usepackage{amsmath, amssymb, latexsym}
\usepackage{enumerate}

\usepackage{float}
\usepackage{xspace}
\newcommand{\Konig}{K\"{o}nig}
\newcommand{\KE}{K\"{o}nig-Egerv\'{a}ry}

\newcommand{\vclp}{{\sc Vertex Cover Above LP}}
\newcommand{\agvc}{{\sc  agvc}}
\newcommand{\agvcfull}{{\sc  Above Guarantee Vertex Cover }}

\newcommand{\runtime}[1]{$O^*(2.3146^{#1})$}

\newtheorem{reduction}{Preprocessing Rule}

\title{Faster Parameterized Algorithms using Linear Programming\thanks{A preliminary version of this paper appears in the proceedings of STACS 2012.}}

\author{Daniel Lokshtanov\thanks{University of California San Diego, San Diego, USA. \texttt{daniello@ii.uib.no}} 
\and N.S. Narayanaswamy\thanks{Department of CSE, IIT Madras, Chennai, India. \texttt{swamy@cse.iitm.ernet.in}} 
\and Venkatesh Raman\thanks{The Institute of Mathematical Sciences, Chennai 600113, India. \newline 
  \texttt{\{vraman|msramanujan|saket\}@imsc.res.in} }
\and \addtocounter{footnote}{-1}M.S. Ramanujan\footnotemark 
\and \addtocounter{footnote}{-1}Saket Saurabh\footnotemark}
\date{}
\begin{document}
\maketitle

\begin{abstract}
We investigate the parameterized complexity of {\sc Vertex Cover} parameterized by the difference between the size of the optimal solution and the value of the linear programming (LP) 
relaxation 
of the problem. By carefully analyzing the change in the LP value in the branching steps, we argue that combining previously known preprocessing rules with the most straightforward branching algorithm yields an $O^*((2.618)^k)$ algorithm for the problem. Here $k$ is the excess of the vertex cover size over the LP optimum, and we write $O^*(f(k))$ for a time complexity of the form $O(f(k)n^{O(1)})$, where $f (k)$ grows exponentially with $k$. We proceed to show that a more sophisticated branching algorithm achieves a runtime of \runtime{k}.

Following this, using known and new reductions, we give \runtime{k} algorithms for the parameterized versions of {\sc Above Guarantee Vertex Cover}, {\sc Odd Cycle Transversal}, {\sc Split Vertex Deletion} and {\sc Almost 2-SAT}, and an $O^*(1.5214^k)$ algorithm for {\sc Ko\"nig Vertex Deletion}, {\sc Vertex Cover Param by OCT} and {\sc Vertex Cover Param by KVD}. These algorithms significantly improve the best known bounds for these problems. The most notable improvement is the new bound for {\sc Odd Cycle Transversal} - this is the first algorithm which beats the dependence on $k$ of the seminal $O^*(3^k)$ algorithm of Reed, Smith and Vetta.
Finally, using our algorithm, we obtain a kernel for the standard parameterization of {\sc Vertex Cover} with at most $2k - O(\log k)$ vertices. Our kernel is simpler than previously known kernels achieving the same size bound. 

{\bf Topics:}  Algorithms and data structures. Graph Algorithms, Parameterized Algorithms.

\end{abstract}

\section{Introduction and Motivation}
In this paper we revisit one of the most studied problems in parameterized complexity, the {\sc Vertex Cover} problem. Given a graph $G=(V,E)$, a subset $S\subseteq V$ is called a {\em vertex cover} if every edge in $E$ has at least one end-point in $S$.  The {\sc Vertex Cover} problem is formally defined as follows. 

\begin{center}
\begin{boxedminipage}{.9\textwidth}

\textsc{Vertex Cover}\vspace{10 pt}

\begin{tabular}{ r l }

\textit{~~~~Instance:} &  An undirected graph $G$ and a positive integer $k$. \\

\textit{Parameter:} & $k$.\\
\textit{Problem:} & Does $G$ have a vertex cover of of size at most $k$? 
 \\
\end{tabular}
\end{boxedminipage}
\end{center}

We start with a few basic definitions regarding parameterized complexity.  For decision problems with input size $n$, 
and a parameter $k$, 
the goal in parameterized complexity is to design an algorithm with runtime $f(k)n^{O(1)}$ where $f$ is a function of $k$ alone, as contrasted with a trivial $n^{k+O(1)}$ algorithm. Problems which admit such algorithms are said to be fixed parameter tractable (FPT). 
The theory of parameterized complexity was developed by Downey and Fellows \cite{DF99}. For recent developments, see the book by Flum and Grohe \cite{FG06}. 

{\sc Vertex Cover}  was one of the first problems that was shown to be FPT~\cite{DF99}. After a long race, the current best algorithm for 
{\sc Vertex Cover} runs in time $O(1.2738^k+kn)$~\cite{ChenKX10}.
However, when $k < m$, the size of the maximum matching, the {\sc Vertex Cover} problem is not interesting, as the answer is trivially NO. Hence, when $m$ is large (for example when the graph has a perfect matching), the running time bound of the standard FPT algorithm is not practical, as $k$, in this case, is quite large.  
This led to the following natural ``above guarantee'' variant of the {\sc Vertex Cover} problem. 

\begin{center}
\begin{boxedminipage}{.9\textwidth}

\textsc{\agvcfull\  (\agvc)}\vspace{10 pt}

\begin{tabular}{ r l }

\textit{~~~~Instance:} &  An undirected graph $G$, a maximum matching $M$ and \\ 
& a positive integer $k$. \\

\textit{Parameter:} & $k-\vert M\vert$.\\
\textit{Problem:} & Does $G$ have a vertex cover of of size at most $k$? 
 \\
\end{tabular}
\end{boxedminipage}
\end{center}
In addition to being a natural parameterization of the classical {\sc Vertex Cover} problem, the \agvc\ problem has a central spot in the ``zoo'' of parameterized problems. We refer to Figure~\ref{fig:zoo} for the details of problems reducing to \agvc. (See the Appendix for the definition of these problems.) In particular an improved algorithm for \agvc\ implies improved algorithms for several other problems as well, including {\sc Almost $2$-SAT} and {\sc Odd Cycle Transversal}. 
\begin{figure}[t]
\begin{center}
\includegraphics[height=145 pt, width=210 pt]{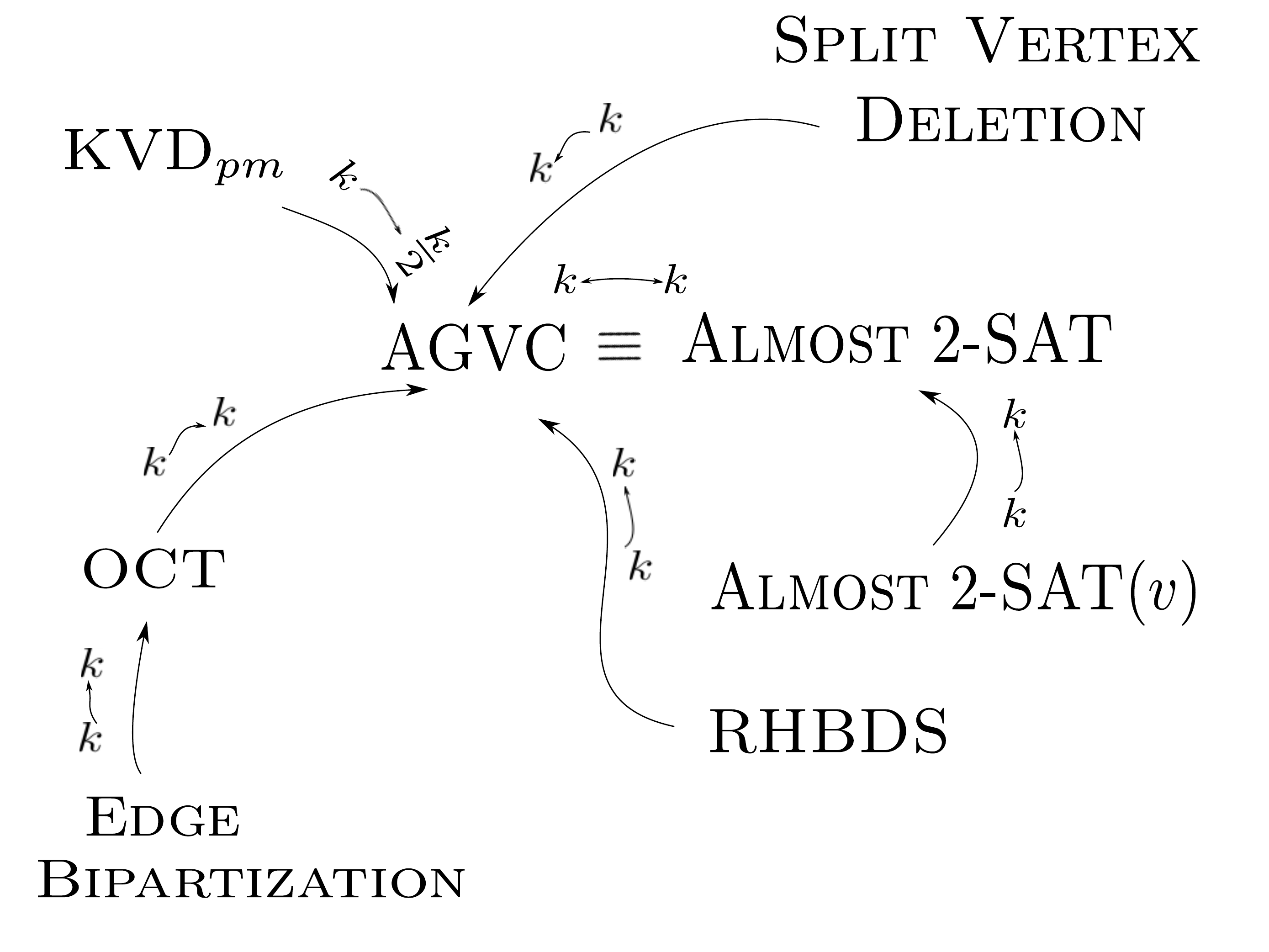}
\caption{The zoo of problems around {\agvc}; An arrow from a problem $P$ to a problem $Q$  indicates that there is a parameterized reduction from $P$ to $Q$ with the parameter changes as indicated on the arrow.}
\label{fig:zoo}
\end{center}
\end{figure}

\agvc\ was first shown fixed parameter tractable by a parameter preserving reduction to {\sc Almost $2$-SAT}. In {\sc Almost $2$-SAT}, we are given a $2$-SAT formula $\phi$, a positive integer $k$ and the objective is to check whether there exists at most $k$ clauses whose deletion from $\phi$ can make the resulting formula satisfiable. The {\sc Almost $2$-SAT} problem was introduced in~\cite{MahajanR99} and a decade later it was proved FPT by Razgon and O'Sullivan~\cite{RazSul09}, who gave a $O^*(15^k)$ time algorithm for the problem. In 2011, there were two new algorithms for the \agvc\ problem~\cite{CPPW11,RRS11}.  The first used new structural results about \KE\  graphs --- graphs where the size of a minimum vertex cover is equal to the size of a maximum matching~\cite{RRS11} while the second invoked a reduction to an ``above guarantee version" of the {\sc Multiway Cut} problem~\cite{CPPW11}. The second algorithm runs in time $O^*(4^k)$ and this is also the  fastest algorithm \agvc\ prior to our work.

In order to obtain the $O^*(4^k)$ running time bound for {\sc Above Guarantee Multiway Cut} (and hence also for \agvc{}), Cygan et al~\cite{CPPW11} introduce a novel measure in terms of which the running time is bounded. Specifically they bound the running time of their algorithm in terms of the difference between the size of the solution the algorithm looks for and the value of the optimal solution to the linear programming relaxation of the problem. Since {\sc Vertex Cover} is a simpler problem than {\sc Multiway Cut} it is tempting to ask whether applying this new approach directly on {\sc Vertex Cover} could yield simpler and faster algorithms for \agvc{}. This idea
is the starting point of our work. 

The well known integer linear programming formulation (ILP)  for {\sc Vertex Cover} is as follows. 
\begin{center}
\begin{boxedminipage}{.9\textwidth}

\textsc{ILP formulation of Minimum Vertex Cover -- ILPVC}\vspace{10 pt}

\begin{tabular}{ r l }

\textit{~~~~Instance:} &   A graph $G=(V,E)$.\\

\textit{Feasible Solution:} &  A function $x: V \rightarrow \{0, 1\}$ satisfying edge constraints \\ 
&  $x(u) + x (v)\geq 1$ for each edge $(u, v) \in E$. \\

\textit{Goal:} & To minimize $w(x )= \Sigma_{u\in V} x (u)$ over all feasible solutions $x$. 
 \\
\end{tabular}
\end{boxedminipage}
\end{center}
In the linear programming relaxation of the above ILP, the constraint $x(v)\in \{0,1\}$ is replaced with 
$ x(v) \geq 0$, for all $v \in V$.  For a graph $G$, we call this relaxation {\sc LPVC($G$)}. Clearly, every integer feasible solution is also a feasible solution 
to {\sc LPVC($G$)}. If the minimum value of {\sc LPVC($G$)} is $vc^*(G)$ then clearly the size of a minimum vertex cover is at least 
$vc^*(G)$. This leads to the following parameterization of {\sc Vertex Cover}.

\begin{center}
\begin{boxedminipage}{.85\textwidth}

\textsc{Vertex Cover above LP}\vspace{10 pt}

\begin{tabular}{ r l }

\textit{~~~~Instance:} &  An undirected graph $G$, positive integers $k$ and $\lceil vc^*(G) \rceil $,  \\ 
& where $vc^*(G)$ is  the minimum value  of {\sc LPVC($G$)}. \\
\textit{Parameter:} & $k-\lceil vc^*(G) \rceil$.\\
\textit{Problem:} & Does $G$ have a vertex cover of of size at most $k$? 
 \\
\end{tabular}
\end{boxedminipage}
\end{center}
Observe that since $vc^*(G) \geq m$, where $m$ is the size of a maximum matching of $G$, we have that $k-vc^*(G) \leq k-m$. Thus, any parameterized algorithm for \textsc{Vertex Cover above LP} is also a parameterized algorithm  for \agvc\ and hence an algorithm for every problem depicted in Figure~\ref{fig:zoo}. 

\begin{table}[t]
\begin{center}
\begin{tabular}{|l|c|c|c|}\hline
Problem Name & Previous $f(k)$/Reference & New $f(k)$ in this paper\\
\hline
\hline
{\sc agvc} & $4^k$ \cite{CPPW11} & $2.3146^k$    \\
\hline
{\sc Almost $2$-SAT} & $4^k$ \cite{CPPW11}& $2.3146^k$   \\
\hline 
{\sc RHorn-Backdoor Detection Set} & $4^k$ \cite{CPPW11,GottlobS08}& $2.3146^k$   \\

\hline
{\sc K\"{o}nig Vertex Deletion} & $4^k$ \cite{CPPW11,MishraRSSS10} & $1.5214^k$  \\
\hline 
{\sc Split Vertex Deletion} & $5^k$ \cite{Cai96} & $2.3146^k$    \\
\hline
{\sc Odd Cycle Transversal} & $3^k$ \cite{ReedSV04} & $2.3146^k$    \\
\hline
{\sc Vertex Cover Param by OCT} & $2^k$ (folklore) & $1.5214^k$  \\
\hline
{\sc Vertex Cover Param by KVD} & -- & $1.5214^k$  \\
\hline
\end{tabular} 
\label{tableofruntime}
\caption{The table gives the previous $f(k)$ bound in the running time of various  problems and the ones obtained in this paper.}
\end{center}
\end{table}

\noindent
{\bf Our Results and Methodology.} We develop a \runtime{(k-vc^*(G))} time branching algorithm for \textsc{Vertex Cover above LP}. In an effort to present the key ideas of our algorithm in as clear a way as possible, we first present a simpler and slightly slower algorithm in Section~\ref{sec:mainalgo}. 
This algorithm exhaustively applies a collection of previously known preprocessing steps. If no further preprocessing is possible the algorithm simply selects an arbitrary 
vertex $v$ and recursively tries to find a vertex cover of size at most $k$ by considering whether $v$ is in the solution or not. 
While the algorithm is simple, the analysis is more involved as it is not obvious that the measure $k-vc^*(G)$ actually drops in the recursive calls. In order to prove that the measure does drop we string together several known results about the linear programming relaxation of {\sc Vertex Cover}, such as 
the classical Nemhauser-Trotter theorem and properties of ``minimum surplus sets''. We find it intriguing that, as our analysis shows, combining well-known reduction rules with naive branching yields fast FPT algorithms for all problems in Figure~\ref{fig:zoo}. 
We then show in Section~\ref{sec:improvalgo} that adding several more involved branching rules to our algorithm yields an improved running time of \runtime{(k-vc^*(G))}. Using this algorithm we obtain even faster algorithms for the problems in Figure~\ref{fig:zoo}. 

We give a list of problems with their previous best running time and the ones obtained in this paper in Table~\ref{tableofruntime}. The most notable among them is the new algorithm for {\sc Odd Cycle Transversal}, the problem of deleting at most $k$ vertices to obtain a bipartite graph. The parameterized complexity of {\sc Odd Cycle Transversal} was a long standing open problem in the area, and only in 2003 Reed et al.~\cite{ReedSV04} developed an algorithm for the problem running in time $O^*(3^k)$. 
However, there has been no further improvement over this algorithm in the last $9$ years; though several reinterpretations of the algorithm have been published~\cite{Huffner09,LokshtanovSS09}. 

We also find the algorithm for {\sc \Konig\ Vertex Deletion}, the problem of deleting at most $k$ vertices to obtain a \Konig\ graph very interesting. {\sc \Konig\ Vertex Deletion} is a natural variant of the odd cycle transversal problem. In~\cite{MishraRSSS10} it was shown that
given a minimum vertex cover one can solve {\sc \Konig\ Vertex Deletion} in polynomial time. In this article we show a relationship between the measure $k-vc^*(G)$ and the minimum number of vertices needed to delete to  obtain a \Konig\ graph. This relationship together with a reduction rule for {\sc \Konig\ Vertex Deletion} based on the Nemhauser-Trotter theorem gives an algorithm for the problem with running time $O^*(1.5124^k)$. 

We also note that using our algorithm, we obtain a polynomial time algorithm for {\sc Vertex Cover} that, given an input $(G,k)$ returns an equivalent instance $(G'=(V',E'),k')$ such that $k'\leq k$ and  $|V(G')|\leq 2k-c \log k$ for any fixed constant $c$. This is known as a kernel for {\sc Vertex Cover} in the literature. We note that this kernel is simpler than another kernel with the same size bound~\cite{Lampis11}.

We hope that this work will lead to a new race towards better algorithms for \textsc{Vertex Cover above LP} like what we have seen for its classical counterpart, {\sc Vertex Cover}.


\section{Preliminaries}
\label{prelim}
For a graph $G=(V,E)$, for a subset $S$ of $V$, the {\it subgraph of $G$ induced by $S$} is denoted by $G[S]$ and it is defined as the subgraph of $G$ with vertex set $S$ and edge set $\{(u,v) \in E :u,v\in S\}$. By $N_G(u)$ we denote the (open) neighborhood of $u$, that is, the set of all vertices adjacent to $u$. Similarly, for a subset $T \subseteq V$, we define $N_G(T)=(\cup_{v\in T} N_G(v))\setminus T$. When it is clear from the context, we drop the subscript $G$ from the notation. We denote by $N_i[S]$, the set $N[N_{i-1}(S)]$ where $N_1[S]=N[S]$, that is, $N_i[S]$ is the set of vertices which are within a distance of $i$ from a vertex in $S$. The surplus of an independent set $X \subseteq V$ is defined as 
${\bf surplus}(X) = |N(X)| - |X|$. 
For a set ${\cal A}$ of independent sets of a graph, ${\bf surplus}({\cal A})= \min_{X\in{\cal A}}{\bf surplus}(X)$.
The surplus of a graph $G$, {\bf surplus($G$)}, is defined to be the minimum surplus over all independent sets in the graph.

By the phrase an optimum solution to LPVC($G$), we mean a feasible solution with $x(v)\geq 0$ for all $v\in V$ 
minimizing the objective function $w(x)=\sum_{u\in V}x(u)$. 
It is well known that for any graph $G$, there exists an optimum solution to LPVC($G$), 
such that $x(u)\in \{0, \frac{1}{2},1\}$ for all $u\in V$~\cite{NT1}. Such a feasible  
optimum solution to LPVC($G$) is called a half integral solution and can be found in polynomial time~\cite{NT1}. 
In this paper we always deal 
with half integral optimum solutions to LPVC($G$). Thus, by default whenever we refer to an {\em optimum solution} to LPVC($G$) we will be referring to a  {\em half integral optimum solution} to LPVC($G$). 
Let $VC(G)$ be the set of all minimum vertex covers of $G$ and $vc(G)$ denote the size of a minimum 
vertex cover of $G$. Let $VC^*(G)$ be the set of all optimal solutions (including non half integral optimal solution) to 
LPVC($G$). By $vc^*(G)$ we denote the value of an optimum solution to LPVC($G$). 
We define $V^x_i=\{u\in V: x(u)=i\}$ for each $i\in \{0,\frac{1}{2},1\}$ and define $x\equiv i$, $i\in \{0,\frac{1}{2},1\}$, if $x(u)=i$ for every $u\in V$. Clearly, $vc(G)\geq vc^*(G)$ and $vc^*(G)\leq \frac{\vert V\vert}{2}$ since $x\equiv \frac{1}{2}$ is always a feasible solution to LPVC($G$).  
We also refer to the $x\equiv \frac{1}{2}$ solution simply as the all $\frac{1}{2}$ solution.

In branching algorithms, we say that a branching step results in a drop of $(p_1, p_2, ...p_l)$  where $p_i, 1 \leq i \leq l$ is an integer, if the measure we use to analyze drops respectively by $p_1, p_2, ... p_l$ in the corresponding branches. We also call the vector $(p_1, p_2, \ldots, p_l)$ the branching vector of the step.
\section{A Simple Algorithm for  {\sc Vertex Cover above LP}}
\label{sec:mainalgo}
In this section, we give a simpler algorithm for \textsc{Vertex Cover above LP}. The algorithm has two phases, a preprocessing phase and a branching phase. We first describe the preprocessing steps used in the algorithm and then give a simple description of the algorithm. Finally, we argue about its correctness and prove the desired running time bound on the algorithm. 
\subsection{Preprocessing}
We describe three standard preprocessing rules to simplify the input instance. We first state the (known) results which allow for their correctness, and then describe the rules.


\begin{lemma}\label{lem:compute good set}{\sc \cite{NT2,Picard:1977fk}}
For a graph $G$, in polynomial time, we can compute an optimal solution $x$ to LPVC($G$) 
such that all $\frac{1}{2}$ is the unique optimal solution to LPVC($G[V^x_{1/2}]$). Furthermore, ${\bf surplus}(G[V^x_{1/2}])>0$.
\end{lemma}

\begin{lemma}
\label{lem:classicalNT}{\sc \cite{NT2}}
Let $G$ be a graph and $x$ be an optimal solution to LPVC($G$). There is a minimum vertex cover for $G$ which contains all the vertices in 
$V^x_1$ and none of the vertices in $V^x_0$. \end{lemma}

\begin{reduction}\label{red:NT_reduction} Apply Lemma~\ref{lem:compute good set} to compute an optimal solution $x$ to LPVC($G$) such that all 
$\frac{1}{2}$ is the unique optimum solution to LPVC($G[V^x_{1/2}]$). Delete the vertices in $V^x_0\cup V^x_1$ from the graph after including $V^x_1$ in the vertex cover we develop, and reduce $k$ by $\vert V^x_1\vert$.
\end{reduction}
\noindent
In the discussions in the rest of the paper, we say that preprocessing rule \ref{red:NT_reduction} applies if all $\frac{1}{2}$ is not the unique solution to LPVC($G$) and that it doesn't apply if all $\frac{1}{2}$ is the unique solution to LPVC($G$).

\noindent
The soundness/correctness of Preprocessing Rule~\ref{red:NT_reduction} follows from Lemma~\ref{lem:classicalNT}. After the application of preprocessing rule \ref{red:NT_reduction}, we know that $x\equiv \frac{1}{2}$ is the unique optimal solution to LPVC() of the resulting graph 
and the graph has a surplus of at least $1$. 


\begin{lemma}
\label{lem:allornonebranch}
{\sc \cite{ChenKX10,NT2}}
Let $G(V,E)$ be a graph, and let $S \subseteq V$ be an independent subset such that ${\bf surplus}(Y) \geq {\bf surplus}(S$) for every set $Y \subseteq S$.
Then there exists a minimum vertex cover for $G$, that contains all of $S$ or none of $S$. In particular, if $S$ is an independent set with the minimum 
surplus, then there exists a minimum vertex cover for $G$, that contains all of $S$ or none of $S$. 
\end{lemma}
\noindent
The following lemma, which handles without branching, the case when the minimum surplus of the graph is $1$, follows from the above lemma.

\begin{lemma}\label{lem:struction}{\sc \cite{ChenKX10,NT2}}
Let $G$ be a graph, and let $Z \subseteq V(G)$ be an independent set such that  ${\bf surplus}(Z)=1$ and for every $Y\subseteq Z$, ${\bf surplus}(Y) \geq {\bf surplus}(Z)$. Then,
\begin{enumerate}
\item If the graph induced by $N(Z)$ is not an independent
set, then there exists a minimum vertex cover in $G$ that
includes all of $N(Z)$ and excludes all of $Z$. 
\item If the graph induced by $N(Z)$ is an independent set, let
$G'$ be the graph obtained from $G$ by removing $Z \cup N (Z)$
and adding a vertex $z$, followed by making $z$ adjacent to every
vertex $v \in G \setminus (Z\cup N(Z))$ which was adjacent to a vertex in $N(Z)$ (also called \emph{identifying} the vertices of $N(Z)$).Then, $G$ has a vertex cover of size at most $k$ if and only if $G'$ has a vertex cover of size at most $k-|Z|$.  
\end{enumerate}
\end{lemma}

\noindent
We now give two preprocessing rules to handle the case when the surplus of the graph is $1$. 
\begin{reduction}
\label{red:edgeinneighbor} 
If there is a set $Z$ such that ${\bf surplus}(Z)=1$ and $N(Z)$ is not an independent set, then we apply Lemma~\ref{lem:struction} to reduce the instance as follows. Include $N(Z)$ in the vertex 
cover, delete $Z\cup N(Z)$ from the graph, and decrease $k$ by $\vert N(Z)\vert$. 
\end{reduction}

\begin{reduction}
\label{red:struction} 
If there is a set $Z$ such that ${\bf surplus}(Z)=1$ and the graph induced by $N(Z)$ is an independent set, then apply Lemma~\ref{lem:struction} to reduce the instance as follows.  
 Remove $Z$ from the graph, identify the vertices of $N(Z)$, and decrease $k$ by $\vert Z\vert$. 
\end{reduction}
\noindent
 The correctness of Preprocessing Rules~\ref{red:edgeinneighbor} and \ref{red:struction}  follows from Lemma~\ref{lem:struction}. The entire preprocessing phase of the algorithm is summarized in Figure~\ref{fig:reductionrules}. Recall that each preprocessing rule can be applied only when none of the preceding rules are applicable, and that Preprocessing rule~\ref{red:NT_reduction} is applicable if and only if there is a solution to LPVC($G$) which does not assign $\frac 1 2$ to every vertex. Hence, when Preprocessing Rule~\ref{red:NT_reduction} does not apply all $\frac 1 2$ is the unique solution for LPVC($G$). We now show that we can test whether Preprocessing Rules ~\ref{red:edgeinneighbor} and ~\ref{red:struction} are applicable on the current instance in polynomial time.
 
 \begin{lemma}\label{lem:test r2}
 Given an instance $(G,k)$ of {\vclp} on which Preprocessing Rule~\ref{red:NT_reduction} does not apply, we can test if Preprocessing Rule~\ref{red:edgeinneighbor} applies on this instance in polynomial time. 
 \end{lemma}
 
 \begin{proof}
 We first prove the following claim.
 
 \begin{claim}
 The graph $G$ (in the statement of the lemma) contains a set $Z$ such that ${\bf surplus}(G)=1$ and $N(Z)$ is not independent if and only if there is an edge $(u,v)\in E$ such that solving LPVC($G$) with $x(u)=x(v)=1$ results in a solution with value exactly $\frac 1 2$ greater than the value of the original LPVC($G$).
 \end{claim}
 
 \begin{proof}
 Suppose there is an edge $(u,v)$ such that $w(x^\prime)=w(x)+{\frac 1 2}$ where $x$ is the solution to the original LPVC($G$) and $x^\prime$ is the solution to LPVC($G$) with $x^\prime(u)=x^\prime(v)=1$ and let $Z=V^{x^\prime}_0$. We claim that the set $Z$ is a set with surplus 1 and that $N(Z)$ is not independent. Since $N(Z)$ contains the vertices $u$ and $v$, $N(Z)$ is not an independent set. Now, since $x \equiv {\frac 1 2}$ (Preprocessing Rule~\ref{red:NT_reduction} does not apply), $w(x^\prime)= w(x)-{\frac 1 2}\vert Z\vert +{\frac 1 2}\vert N(Z)\vert=w(x)+{\frac 1 2}$. Hence, $\vert N(Z)\vert-\vert Z\vert={\bf surplus}(Z)=1$.
 
Conversely, suppose that there is a set $Z$ such that ${\bf surplus}(Z)=1$ and $N(Z)$ contains vertices $u$ and $v$ such that $(u,v)\in E$. Let $x^\prime$ be the assignment which assigns 0 to all vertices of $Z$, 1 to $N(Z)$ and $\frac 1 2$ to the rest of the vertices. Clearly, $x^\prime$ is a feasible assignment and $w(x^\prime)=\vert N(Z)\vert + {\frac 1 2}\vert V\setminus (Z\cup N(Z))\vert$. Since Preprocessing Rule~\ref{red:NT_reduction}  does not apply, $w(x^\prime)-w(x)=\vert N(Z)\vert -{\frac 1 2}(\vert Z\vert +\vert N(Z)\vert)={\frac 1 2}(\vert N(Z)\vert -\vert Z\vert)={\frac 1 2}$, which proves the converse part of the claim.
 
 \end{proof}

Given the above claim, we check if Preprocessing Rule~\ref{red:edgeinneighbor} applies by doing the following for every edge $(u,v$) in the graph.

Set $x(u)=x(v)=1$ and solve the resulting LP looking for a solution whose optimum value is exactly ${\frac 1 2}$ more than the optimum value of LPVC($G$).

\end{proof}

 \begin{lemma}\label{lem:test r3}
 Given an instance $(G,k)$ of {\vclp} on which Preprocessing Rules~\ref{red:NT_reduction}  and ~\ref{red:edgeinneighbor} do not apply, we can test if Preprocessing Rule~\ref{red:struction} applies on this instance in polynomial time. 
 \end{lemma}
 
 \begin{proof}
 We first prove a claim analogous to that proved in the above lemma.
  \begin{claim}
 The graph $G$ (in the statement of the lemma) contains a set $Z$ such that ${\bf surplus}(G)=1$ and $N(Z)$ is independent if and only if there is a vertex $u\in V$ such that solving LPVC($G$) with $x(u)=0$ results in a solution with value exactly $\frac 1 2$ greater than the value of the original LPVC($G$).
 \end{claim}
 
 \begin{proof}
 Suppose there is a vertex $u$ such that $w(x^\prime)=w(x)+{\frac 1 2}$ where $x$ is the solution to the original LPVC($G$) and $x^\prime$ is the solution to LPVC($G$) with $x^\prime(u)=0$ and let $Z=V^{x^\prime}_0$. We claim that the set $Z$ is a set with surplus 1 and that $N(Z)$ is independent. Since $x \equiv {\frac 1 2}$ (Preprocessing Rule~\ref{red:NT_reduction} does not apply), $w(x^\prime)= w(x)-{\frac 1 2}\vert Z\vert +{\frac 1 2}\vert N(Z)\vert=w(x)+{\frac 1 2}$. Hence, $\vert N(Z)\vert-\vert Z\vert={\bf surplus}(Z)=1$. Since Preprocessing Rule~\ref{red:edgeinneighbor} does not apply, it must be the case that $N(Z)$ is independent.
 
Conversely, suppose that there is a set $Z$ such that ${\bf surplus}(Z)=1$ and $N(Z)$ is independent . Let $x^\prime$ be the assignment which assigns 0 to all vertices of $Z$ and 1 to all vertices of $N(Z)$ and $\frac 1 2$ to the rest of the vertices. Clearly, $x^\prime$ is a feasible assignment and $w(x^\prime)=\vert N(Z)\vert + {\frac 1 2}\vert V\setminus (Z\cup N(Z))\vert$. Since Preprocessing Rule~\ref{red:NT_reduction}  does not apply, $w(x^\prime)-w(x)=\vert N(Z)\vert -{\frac 1 2}(\vert Z\vert +\vert N(Z)\vert)={\frac 1 2}(\vert N(Z)\vert -\vert Z\vert)={\frac 1 2}$. This proves the converse part of the claim with $u$ being any vertex of $Z$. 
 
 \end{proof}

Given the above claim, we check if Preprocessing
Rule~\ref{red:struction} applies by doing the following for every vertex $u$ in the graph.

Set $x(u)=0$, solve the resulting LP and look for a solution whose optimum value exactly $\frac 1 2$ more than the optimum value of $LPVC(G)$.

\end{proof}


\begin{figure}[t]
\begin{center}
\begin{boxedminipage}{.9\textwidth}
{\small
The rules are applied in the order in which they are presented, that is, any rule is applied only when none of the earlier
rules are applicable.
\begin{description}

\item
{\bf Preprocessing rule} \ref{red:NT_reduction}: Apply Lemma~\ref{lem:compute good set} to compute an optimal solution $x$ to LPVC($G$) such that all 
$\frac{1}{2}$ is the unique optimum solution to LPVC($G[V^x_{1/2}]$). Delete the vertices in $V^x_0\cup V^x_1$ from the graph after including $V^x_1$ in the vertex cover we develop, and reduce $k$ by $\vert V^x_1\vert$.

\item
{\bf Preprocessing rule \ref{red:edgeinneighbor}}: 
Apply Lemma~\ref{lem:test r2} to test if there is a set $Z$ such that ${\bf surplus}(Z)=1$ and $N(Z)$ is not an independent set. If such a set does exist, then we apply Lemma~\ref{lem:struction} to reduce the instance as follows. Include $N(Z)$ in the vertex 
cover, delete $Z\cup N(Z)$ from the graph, and decrease $k$ by $\vert N(Z)\vert$. 

\item
{\bf Preprocessing rule \ref{red:struction}}: 
Apply Lemma~\ref{lem:test r3} to test if there is a set $Z$ such that ${\bf surplus}(Z)=1$ and $N(Z)$ is an independent set. If there is such a set $Z$ then apply Lemma~\ref{lem:struction} to reduce the instance as follows.  
Remove $Z$ from the graph, identify the vertices of $N(Z)$, and decrease $k$ by $\vert Z\vert$. 
\end{description}

}
\end{boxedminipage}
\end{center}
\caption{\bf \em Preprocessing Steps}
\label{fig:reductionrules}
\end{figure}

\begin{define}
For a graph $G$, we denote by ${\cal R}(G)$ the graph obtained after applying 
Preprocessing Rules~\ref{red:NT_reduction},  \ref{red:edgeinneighbor} and \ref{red:struction} exhaustively in this order. 
\end{define}
\noindent


Strictly speaking ${\cal R}(G)$ is not a well defined function since the reduced graph could depend on which sets the reduction rules are applied on, and these sets, in turn, depend on the solution to the LP. To overcome this technicality we let ${\cal R}(G)$ be a function not only of the graph $G$ but also of the representation of $G$ in memory. Since our reduction rules are deterministic (and the LP solver we use as a black box is deterministic as well), running the reduction rules on (a specific representation of) $G$ will always result in the same graph, making the function ${\cal R}(G)$ well defined. Finally, observe that for any $G$ the all $\frac{1}{2}$ is the unique optimum solution to the LPVC(${\cal R}(G)$) and ${\cal R}(G)$ has a surplus of at least $2$.


\subsection{Branching}
After the preprocessing rules are applied exhaustively, we pick an arbitrary vertex $u$ in the graph and branch on it. 
In other words, in one branch, we add $u$ into the vertex cover, decrease $k$ by 1, and delete $u$ from the graph, and in the other branch, we add $N(u)$ into the vertex cover, decrease $k$ by $\vert N(u)\vert$, and delete $\{u\}\cup N(u)$ from the graph. The correctness of this algorithm follows from the soundness of the preprocessing rules and the fact that the branching is exhaustive.

\subsection{Analysis}
\label{section:analysissimple}
In order to analyze the running time of our algorithm, we define a measure $\mu=\mu(G,k)=k-vc^*(G)$. We first show that our preprocessing rules do not increase this measure. Following this, we will prove a lower bound on the decrease in the measure occurring as a result of the branching, thus allowing us to bound the running time of the algorithm in terms of the measure $\mu$. For each case, we let $(G^\prime,k^\prime)$ be the instance resulting by the application of the rule or branch, and let $x^\prime$ be an optimum solution to LPVC($G^\prime$).



\begin{enumerate}
\setlength{\itemsep}{-2pt}
\item

Consider the application of Preprocessing Rule~\ref{red:NT_reduction}.  We know that $k^\prime=k-\vert V^x_1\vert$. Since $x^\prime\equiv \frac{1}{2}$ is the unique optimum solution to LPVC($G^\prime$), and $G^\prime$ comprises precisely the vertices of $V^x_{1/2}$, the value of the optimum solution to LPVC($G^\prime$) is exactly $\vert V^x_1\vert$ less than that of $G$. Hence, $\mu(G,k)=\mu(G^\prime,k^\prime)$.


\item We now consider the application of Preprocessing Rule~\ref{red:edgeinneighbor}. We know that $N(Z)$ was not independent. In this case, 
$k^\prime=k-\vert N(Z)\vert$. We also know that 
$w(x^\prime)= \sum_{u \in V} x'(u) = w(x)-\frac{1}{2}(\vert Z\vert+\vert N(Z)\vert) +\frac{1}{2}(\vert V^{x^\prime}_1\vert-\vert V^{x^\prime}_0\vert)$.  
Adding and subtracting $\frac{1}{2}(\vert N(Z)\vert)$, we get $w(x^\prime)=w(x)-\vert N(Z)\vert +\frac{1}{2}(\vert N(Z)\vert-\vert Z\vert) +\frac{1}{2}(\vert V^{x^\prime}_1\vert-\vert V^{x^\prime}_0\vert)$.  But, $Z\cup V^{x^\prime}_0$ is an independent set in $G$, and $N(Z\cup V^{x^\prime}_0)=N(Z)\cup V^{x^\prime}_1$ in $G$. Since ${\bf surplus}(G)\geq 1$, $\vert N(Z\cup V^{x^\prime}_0)\vert -\vert Z\cup V^{x^\prime}_0\vert \geq 1$. Hence,
$w(x^\prime)=w(x)-\vert N(Z)\vert +\frac{1}{2}(\vert N(Z\cup V^{x^\prime}_0)\vert -\vert Z\cup V^{x^\prime}_0\vert)\geq w(x)- \vert N(Z)\vert +\frac{1}{2}$. Thus, $\mu(G^\prime,k^\prime)\leq \mu(G,k)-\frac{1}{2}$.

\item We now consider the application of Preprocessing Rule~\ref{red:struction}. We know that $N(Z)$ was independent. In this case, $k^\prime=k-\vert Z\vert$. We claim that $w(x^\prime)\geq w(x)-\vert Z\vert$. Suppose that this is not true. Then, it must be the case that $w(x^\prime)\leq w(x)-\vert Z\vert-\frac{1}{2} $. We will now consider three cases depending on the value $x^\prime(z)$ where $z$ is the vertex in $G^\prime$ resulting from the identification of $N(Z)$.

\noindent 
{\bf Case 1:} $x^\prime(z)=1$. Now consider the following function $x^{\prime\prime}:V\rightarrow \{0,\frac{1}{2},1\}$. For every vertex $v$ in $G^\prime\setminus \{z\}$, retain the value assigned by $x^\prime$, that is $x^{\prime\prime}(v)=x^\prime(v)$. For every vertex in $N(Z)$, assign 1 and for every vertex in $Z$, assign 0. Clearly this is a feasible solution. 
But now, $w(x^{\prime\prime})= w(x^\prime)-1+|N(Z)| = w(x^\prime)-1+(\vert Z\vert +1)\leq w(x)-\frac{1}{2}$. Hence, we have a feasible solution of value less than the optimum, which is a contradiction. 

\noindent 
{\bf Case 2:} $x^\prime(z)=0$. Now consider the following function $x^{\prime\prime}:V\rightarrow \{0,\frac{1}{2},1\}$. For every vertex $v$ in $G^\prime\setminus \{z\}$, retain the value assigned by $x^\prime$, that is $x^{\prime\prime}(v)=x^\prime(v)$. For every vertex in $Z$, assign 1 and for every vertex in $N(Z)$, assign 0. Clearly this is a feasible solution. 
But now, $w(x^{\prime\prime})= w(x^\prime)+\vert Z\vert\leq w(x)-\frac{1}{2}$. Hence, we have a feasible solution of value less than the optimum, which is a contradiction. 

\noindent 
{\bf Case 3:} $x^\prime(z)=\frac{1}{2}$. Now consider the following function $x^{\prime\prime}:V\rightarrow \{0,\frac{1}{2},1\}$. For every vertex $v$ in $G^\prime\setminus \{z\}$, retain the value assigned by $x^\prime$, that is $x^{\prime\prime}(v)=x^\prime(v)$. For every vertex in $Z\cup N(Z)$, assign $\frac{1}{2}$. Clearly this is a feasible solution.
But now, $w(x^{\prime\prime})= w(x^\prime)-\frac{1}{2}+\frac{1}{2}(\vert Z\vert +\vert N(Z)\vert)=w(x^\prime)-\frac{1}{2}+\frac{1}{2}(\vert Z\vert+\vert Z\vert +1)\leq w(x)-\frac{1}{2}$. Hence, we have a feasible solution of value less than the optimum, which is a contradiction.

Hence, $w(x^\prime)\geq w(x)-\vert Z\vert$, which implies that $\mu(G^\prime,k^\prime)\leq \mu(G,k)$.


\item 
We now consider the branching step.
\begin{enumerate}

\item Consider the case when we pick $u$ in the vertex cover. In this case, $k^\prime=k-1$. We claim that $w(x^\prime)\geq w(x)-\frac{1}{2}$. Suppose that this is not the case. Then, it must be the case that $w(x^\prime)\leq w(x)-1$. Consider the following assignment $x^{\prime\prime}:V\rightarrow \{0,\frac{1}{2},1\}$ to LPVC($G$). For every vertex $v\in V\setminus \{u\}$, set $x^{\prime\prime}(v)=x^\prime(v)$ and set $x^{\prime\prime}(u)=1$. Now, $x^{\prime\prime}$ is clearly a feasible solution and has a value at most that of $x$. But this contradicts our assumption that $x\equiv \frac{1}{2}$ is the unique optimum solution to LPVC($G$). Hence, $w(x^\prime)\geq w(x)-\frac{1}{2}$, which implies that $\mu(G^\prime,k^\prime)\leq \mu(G,k)-\frac{1}{2}$.


\item Consider the case when we don't pick $u$ in the vertex cover. In this case, $k^\prime=k-\vert N(u)\vert$.
%
We know that
$w(x^\prime)=w(x)-\frac{1}{2}(\vert \{u\}\vert+\vert N(u)\vert) +\frac{1}{2}(\vert V^{x^\prime}_1\vert-\vert V^{x^\prime}_0\vert)$.  Adding and subtracting $\frac{1}{2}(\vert N(u)\vert)$, 
we get $w(x^\prime)=w(x)-\vert N(u)\vert -\frac{1}{2}(\vert \{u\}\vert-\vert N(u)\vert) +\frac{1}{2}(\vert V^{x^\prime}_1\vert-\vert V^{x^\prime}_0\vert)$. 
But, $\{u\}\cup V^{x^\prime}_0$ is an independent set in $G$, and $N(\{u\}\cup V^{x^\prime}_0)=N(u)\cup V^{x^\prime}_1$ in $G$. Since 
${\bf surplus}(G)\geq 2$, $\vert N(\{u\}\cup V^{x^\prime}_0)\vert -\vert \{u\}\cup V^{x^\prime}_0\vert \geq 2$. 
Hence, $w(x^\prime)=w(x)-\vert N(u)\vert +\frac{1}{2}(\vert N(\{u\}\cup V^{x^\prime}_0)\vert -\vert \{u\}\cup V^{x^\prime}_0\vert)\geq w(x)- \vert N(u)\vert +1.$

Hence, $\mu(G^\prime,k^\prime)\leq \mu(G,k)-1$.
\end{enumerate}
\end{enumerate}
\noindent
We have thus shown that the preprocessing rules do not increase the measure $\mu=\mu(G,k)$ and the branching step results in a $(\frac{1}{2}, 1)$ branching vector, resulting in the recurrence $T(\mu) \leq T(\mu-\frac{1} {2}) + T(\mu-1)$  which solves to 
$(2.6181)^{\mu} = (2.6181)^{k-vc^*(G)}$. Thus, we get a 
$(2.6181)^{(k-vc^*(G))}$ algorithm for {\sc Vertex Cover above LP}. 


\begin{theorem}\label{thm:main theorem}
{\sc Vertex Cover above LP} can be solved in time $O^*({(2.6181)}^{k-vc^*(G)})$.
\end{theorem}
\noindent
By applying the above theorem iteratively for increasing values of $k$, we can compute a minimum vertex cover of $G$ and hence we have the following corollary.

\begin{corollary}\label{cor:compute min vc}
There is an algorithm that, given a graph $G$, computes a minimum vertex cover of $G$ in time $O^*(2.6181^{(vc(G)-vc^*(G))})$.
\end{corollary}

\section{Improved Algorithm for  {\sc Vertex Cover above LP}}
\label{sec:improvalgo}
In this section we give an improved algorithm for  {\sc Vertex Cover above LP} using some more branching steps based on the structure of the neighborhood of 
the vertex (set) on which we branch. The goal is to achieve branching vectors better that $({\frac 1 2},1)$.

\subsection{Some general claims to measure the drops}



First, we capture the drop in the measure in the branching steps, including when we branch on a larger sized sets. In particular, when we branch on a set $S$ of vertices, in one branch we set all vertices of $S$ to $1$, and in the other, we set all vertices of $S$ to $0$. 
Note, however that such a branching on $S$ may not be exhaustive (as the branching doesn't explore the possibility that some vertices of $S$ are set to $0$ and some are set to $1$) unless the set $S$ satisfies the premise of Lemma \ref{lem:allornonebranch}. Let $\mu=\mu(G,k)$ be the measure as defined in the previous section.

\begin{lemma}
\label{lemma:maindrop}
Let $G$ be a graph with ${\bf surplus}(G)=p$, and let $S$ be an independent 
set. Let ${\cal H}_S$ be the collection of all independent sets of $G$ that contain $S$ (including $S$).  Then, including $S$ in the 
vertex cover while branching leads to a decrease of $\min\{\frac{\vert S \vert}{2},\frac{p}{2}\}$ in $\mu$; 
and the branching excluding $S$ from the vertex cover leads to a drop of 
$\frac{{\bf surplus}({\cal H}_S)}{2}\geq \frac{p}{2}$ in $\mu$. 
\end{lemma}
\begin{proof} Let $(G^\prime,k^\prime)$ be the instance resulting from the branching, and let $x^\prime$ be an optimum solution to LPVC($G^\prime$).
Consider the case when we pick $S$ in the vertex cover. In this case, $k^\prime=k-\vert S \vert$. We know that
$w(x^\prime)=w(x)-\frac{\vert S\vert}{2}+\frac{1}{2}(\vert V^{x^\prime}_1\vert-\vert V^{x^\prime}_0\vert)$. If 
$V^{x^\prime}_0=\emptyset$, then we know that $V^{x^\prime}_1=\emptyset$, and hence we have that $w(x^\prime)=w(x)-\frac{\vert S\vert}{2}$. Else, by 
adding and subtracting $\frac{1}{2}(\vert S \vert)$, we get 
$w(x^\prime)=w(x)-\vert S \vert +\frac{\vert S\vert}{2} +\frac{1}{2}(\vert V^{x^\prime}_1\vert-\vert V^{x^\prime}_0\vert)$. 
However, $N(V^{x^\prime}_0)\subseteq  S\cup V^{x^\prime}_1$ in $G$. 
Thus, $w(x^\prime) \geq w(x)-\vert S \vert  +\frac{1}{2}(\vert N(V^{x^\prime}_0)\vert-\vert V^{x^\prime}_0\vert)$. 
We also know that $V^{x^\prime}_0$ is an independent set in $G$, and thus,  
$\vert N(V^{x^\prime}_0)\vert -\vert V^{x^\prime}_0\vert \geq {\bf surplus}(G)=p$.  Hence, in the first case  
$\mu(G^\prime,k^\prime)\leq \mu(G,k)-\frac{\vert S\vert}{2}$ and in the second case $\mu(G^\prime,k^\prime)\leq \mu(G,k)-\frac{p}{2}$. 
Thus, the drop in the measure when $S$ is included in the vertex cover is at least $\min\{\frac{\vert S \vert}{2},\frac{p}{2}\}$. 

Consider the case when we don't pick $S$ in the vertex cover. In this case, $k^\prime=k-\vert N(S)\vert$. We know that
$w(x^\prime)=w(x)-\frac{1}{2}(\vert S\vert+\vert N(S)\vert) +\frac{1}{2}(\vert V^{x^\prime}_1\vert-\vert V^{x^\prime}_0\vert)$.  
Adding and subtracting $\frac{1}{2}(\vert N(S)\vert)$, 
we get $w(x^\prime)=w(x)-\vert N(S)\vert +\frac{1}{2}(\vert N(S)\vert-\vert S\vert) +\frac{1}{2}(\vert V^{x^\prime}_1\vert-\vert V^{x^\prime}_0\vert)$. 
But, $S\cup V^{x^\prime}_0$ is an independent set in $G$, and $N(S\cup V^{x^\prime}_0)=N(S)\cup V^{x^\prime}_1$ in $G$. Thus, 
$\vert N(S\cup V^{x^\prime}_0)\vert -\vert S\cup V^{x^\prime}_0\vert \geq {\bf surplus}({\cal H}_S)$.   
Hence, $w(x^\prime)=w(x)-\vert N(S)\vert +\frac{1}{2}(\vert N(S\cup V^{x^\prime}_0)\vert -\vert S\cup V^{x^\prime}_0\vert)\geq w(x)- \vert N(S)\vert +\frac{{\bf surplus}({\cal H}_S)}{2}$. Hence, $\mu(G^\prime,k^\prime)\leq \mu(G,k)-\frac{{\bf surplus}({\cal H}_S)}{2}$. 
\end{proof} 
\noindent
Thus, after the preprocessing steps (when the surplus of the graph is at least $2$), suppose we manage to find (in polynomial time) a set $S \subseteq V$ such that 
\begin{itemize}
\item
{\bf surplus}$(G)$ = {\bf surplus}$(S)$ ={\bf surplus}$({\cal H}_S)$, 
\item
$|S| \geq 2$, and 
\item
that the branching that sets all of $S$ to $0$ or all of $S$ to $1$ is exhaustive. 
\end{itemize}
Then, Lemma \ref{lemma:maindrop} guarantees that branching on this set right away leads to a $(1,1)$ branching vector.
We now explore the cases in which such sets do exist . Note that the first condition above implies the third from the Lemma \ref{lem:allornonebranch}. First, we show that if there exists a set $S$ such that $|S| \geq 2$ and 
{\bf surplus} $(G)$ = {\bf surplus}$(S)$, then we can find such a set in polynomial time.

\begin{lemma}
\label{lem:largesurplus}
Let $G$ be a graph on which  Preprocessing Rule~\ref{red:NT_reduction} does not apply (i.e. all $\frac{1}{2}$ is the unique solution to LPVC(G)).
If $G$ has an independent set $S'$ such that $|S'|\geq 2$ and ${\bf surplus}(S')={\bf surplus}(G)$,  
then in polynomial time we can find an independent set $S$ such that $|S|\geq 2$ and ${\bf surplus}(S)={\bf surplus}(G)$. 
\end{lemma}
\begin{proof}
By our assumption we know that $G$ has an independent set $S'$ such that $|S'|\geq 2$ and 
${\bf surplus}(S')={\bf surplus}(G)$. Let $u, v \in S'$.
Let ${\cal H}$ be the collection of all independent sets of $G$ containing $u$ and $v$. 
Let $x$ be an optimal solution to LPVC($G$) obtained after setting $x(u)=0$ and 
$x(v)=0$. Take $S=V_{0}^x$, clearly, we have that $\{u,v\}\subseteq V_0^x$.  
We now have  the following claim. 
\begin{claim}
\label{claim:computesuv}
{\bf surplus}$(S)$ = {\bf surplus}$(G)$.
\end{claim}
\begin{proof}
We know that the objective value of LPVC($G$) after setting $x(u)=x(v)=0$, 
$w(x) = |V|/2 + (|N(S)| - |S|)/2 = |V|/2 + {\bf surplus}(S)/2$, as
 all $\frac{1}{2}$ is the unique solution to LPVC($G$).

Another solution $x'$, for LPVC($G$) that sets $u$ and $v$ to $0$, is obtained by setting $x'(a) =0$ for every $a \in S'$, $x'(a) =1$ for every $a \in N(S')$ and by setting all other variables to $1/2$. It is easy to see that such a solution is a feasible solution of the required kind and $w(x') = |V|/2 + (|N(S')| - |S'|)/2 = |V|/2 + {\bf surplus} (S')/2$. However, as $x$ is also an optimum solution, $w(x)=w(x^\prime)$, and hence we have that ${\bf surplus (S)} \leq {\bf surplus (S')}$. But as $S'$ is a set of minimum surplus in $G$, we have that ${\bf surplus} (S) = {\bf surplus} (S') = {\bf surplus} (G)$ proving the claim.
\end{proof}
\noindent
Thus, 
we can find a such a set $S$ in polynomial time by solving
LPVC($G$) after setting $x(u)=0$ and $x(v)=0$ for every pair of vertices $u, v$ such that $(u, v) \notin E$ and picking 
that set $V_0^x$
which has the minimum surplus among 
all $x's$ among all pairs $u,v$.
Since any $V_0^x$ contains at least 2 vertices, we have that $|S|\geq 2$.

\end{proof}

\subsection{(1,1) drops in the measure}
Lemma \ref{lemma:maindrop} and Lemma \ref{lem:largesurplus} together imply that, if there is a minimum surplus set of size at least $2$ in the graph, then we can find and branch on that set to get a $(1,1)$ drop in the measure. 

Suppose that there is no minimum surplus set of size more than $1$. Note that, by Lemma \ref{lemma:maindrop}, when ${\bf surplus} (G) \geq 2$, we get a drop of $({\bf surplus(G)})/2 \geq 1$ in the branch where we \emph{exclude} a vertex or a set. Hence, if we find some vertex (set) to exclude in either branch of a two way branching, we get a $(1,1)$ branching vector. We now identify another such case.

\begin{lemma}
Let $v$ be a vertex such that $G[N(v)\setminus \{u\}]$ is a clique for some neighbor $u$ of $v$. Then, there exists a minimum vertex cover 
that doesn't contain $v$ or doesn't contain $u$.
\end{lemma} 
\begin{proof}
Towards the proof we first show the following well known observation. 
\begin{claim}
\label{claim:lembranch2}
 Let $G$ be a graph and $v$ be a vertex. Then there exists a a minimum vertex cover for $G$ containing 
$N(v)$ or at most $|N(v)|-2$ vertices from $N(v)$.
\end{claim}
\begin{proof}
If a minimum vertex cover of $G$, say $C$, contains exactly $|N(v)|-1$ 
vertices of $N(v)$, then we know that $C$ must contain $v$. Observe that $C^\prime=C\setminus \{v\} \cup N(v)$ is also a vertex cover of $G$ of the same size as $C$.
However, in this case, we have a minimum vertex cover containing $N(v)$. 
Thus, there exists a minimum vertex cover of $G$ containing 
$N(v)$ or at most $|N(v)|-2$ vertices from $N(v)$.
\end{proof}

Let $v$ be a vertex such that $G[N(v)\setminus \{u\}]$ is a clique. Then, branching 
on $v$ would imply that in one branch we are excluding $v$ from the vertex cover and in the other 
we are including $v$. Consider the branch where we include $v$ in the vertex cover.  Since $G[N(v)\setminus \{u\}]$ is a clique 
we have to pick at least $|N(v)|-2$ vertices from $G[N(v)\setminus \{u\}]$.  Hence, by Claim~\ref{claim:lembranch2}, we can assume that the vertex $u$ is not part of the vertex cover. 
This completes the proof. 
\end{proof}
\noindent
Next, in order to identify another case where we might obtain a $(1,1)$ branching vector, we first observe and capture the fact that 
when Preprocessing Rule~\ref{red:edgeinneighbor} is applied, the measure
$k-vc^*(G)$ actually drops by at least ${1 \over 2}$ (as proved in item 2 of the analysis of the simple algorithm in Section \ref{section:analysissimple}).
\begin{lemma}
\label{lem:prule2measuredrop}
Let $(G,k)$ be the input instance and  $(G^\prime,k^\prime)$ be the instance obtained after applying  
Preprocessing Rule~\ref{red:edgeinneighbor}. Then, $\mu(G^\prime,k^\prime)\leq \mu(G,k)-\frac{1}{2}$.
\end{lemma}
\noindent
Thus, after we branch on an arbitrary vertex, if we are able to apply Preprocessing Rule~\ref{red:edgeinneighbor} in the branch where we include that vertex, we get a $(1,1)$ drop. For, in the branch where we exclude the vertex, we get a drop of $1$ by Lemma~\ref{lemma:maindrop}, and in the branch where we include the vertex, we get a drop of $\frac{1}{2}$ by Lemma \ref{lemma:maindrop}, which is then followed by a drop of $\frac{1}{2}$ due to Lemma \ref{lem:prule2measuredrop}. 

Thus, after preprocessing, the algorithm performs the following steps (see Figure~\ref{fig:algoimprovedbranching}) each of which results in a $(1,1)$ drop as argued before. Note that Preprocessing Rule~\ref{red:NT_reduction} cannot apply in the graph $G\setminus \{v\}$ since the surplus of $G$ can drop by at most 1 by deleting a vertex. Hence, checking if rule {\bf B}\ref{bran3} applies is equivalent to checking if, for some vertex $v$, Preprocessing Rule~\ref{red:edgeinneighbor} applies in the graph $G\setminus \{v\}$. Recall that, by Lemma~\ref{lem:test r2} we can check this in polynomial time and hence we can check if {\bf B}\ref{bran3} applies on the graph in polynomial time.

\begin{figure}[h]
\begin{center}
\begin{boxedminipage}{.9\textwidth}
{\small


\noindent 
{\bf Branching Rules.}\\
These branching rules are applied in this order. 

\begin{brule}
\label{bran1}
{\rm
Apply Lemma~\ref{lem:largesurplus} to test if there is a set $S$ such that \textbf{surplus}$(S)$=\textbf{surplus}$(G)$ and $|S|\geq 2$. If so, then branch on $S$. }
\end{brule}

\begin{brule}
\label{bran2}
{\rm
Let $v$ be a vertex such that $G[N(v)\setminus \{u\}]$ is a clique for some vertex $u$ in $N(v)$. Then in one branch 
add $N(v)$ into the vertex cover, decrease $k$ by $\vert N(v)\vert$, and 
delete $N[v]$ from the graph. In the other branch add $N(u)$ into the vertex cover, decrease $k$ by 
$\vert N(u)\vert$, and delete $N[u]$ from the graph.}
\end{brule}

\begin{brule}
\label{bran3}
{\rm
Apply Lemma~\ref{lem:test r2} to test if there is a vertex $v$ such that preprocessing Rule~\ref{red:edgeinneighbor} applies in $G\setminus \{v\}$. If there is such a vertex, then branch on $v$. }
\end{brule}
}
\end{boxedminipage}
\end{center}
\caption{\bf \em Outline of the branching steps yielding $(1,1)$ drop.}
\label{fig:algoimprovedbranching}
\end{figure}

\subsection{A Branching step yielding $(1/2, 3/2)$ drop}

Now, suppose none of the preprocessing and branching rules presented thus far apply.
Let $v$ be a vertex with degree at least $4$. Let $S=\{ v \}$ and recall that ${\cal H}_S$ 
is the collection of all independent sets containing $S$, and {\bf surplus} (${\cal H}_S$) 
is an independent set with minimum surplus in ${\cal H}_S$. We claim that ${\bf surplus}({\cal H}_S) \geq 3$.

As the preprocessing rules don't apply, clearly ${\bf surplus}({\cal H}_S) \geq {\bf surplus}(G) \geq 2$. If ${\bf surplus}({\cal H}_S) =2$, then the set that realizes ${\bf surplus}({\cal H}_S)$ is not $S$ (as the ${\bf surplus}(S) = degree (v) -1 =3$), but a superset of $S$, which is of cardinality at least $2$. Then, the branching rule {\bf B}\ref{bran1} would have applied which is a contradiction. This proves the claim. 
Hence by Lemma \ref{lemma:maindrop}, we get a drop of at least $3/2$ in the branch that excludes the vertex $v$ resulting in a $(1/2, 3/2)$ drop. This branching step is presented in Figure~\ref{fig:algoimprovedb4}.

\begin{figure}[h]
\begin{center}
\begin{boxedminipage}{.9\textwidth}
{\small

\begin{brule}
\label{bran4}
{\rm
If there exists a vertex $v$ of degree at least $4$ then branch on $v$. }
\end{brule}
 
}
\end{boxedminipage}
\end{center}
\caption{\bf \em The branching step yielding a $(1/2,3/2)$ drop.}
\label{fig:algoimprovedb4}
\end{figure}

\subsection{A Branching step yielding $(1, 3/2, 3/2)$ drop}

Next, we observe that when branching on a vertex, if in the branch that includes the vertex in the vertex cover (which guarantees a drop of $1/2$), any of the branching rules {\bf B}\ref{bran1} or {\bf B}\ref{bran2} or {\bf B}\ref{bran3} applies, then combining the subsequent branching with this branch of the current branching step results in a net drop of $(1, 3/2, 3/2)$ (which is $(1, 1/2+1, 1/2+1)$) (see Figure~\ref{fig:B5B6}~(a)). Thus, we add the following branching rule to the algorithm (Figure~\ref{fig:algoimprovedb5}).\\

\begin{figure}[t]
\begin{center}
\includegraphics[height=140 pt, width=300 pt]{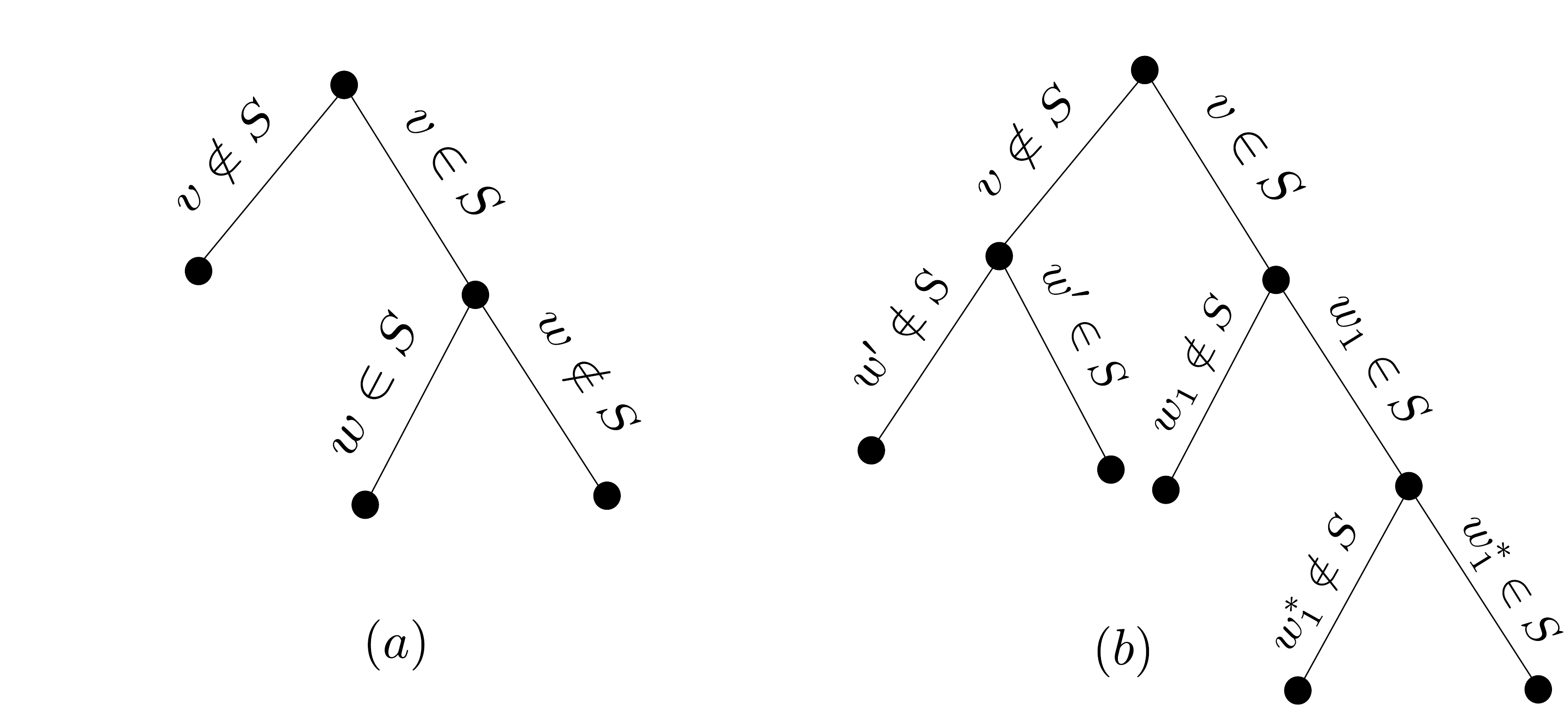}
\caption{Illustrations of the branches of rules (a) {\bf B}\ref{bran5} and (b) {\bf B}\ref{bran6} }
\label{fig:B5B6}
\end{center}
\end{figure}

\begin{figure}[H]
\begin{center}
\begin{boxedminipage}{.9\textwidth}
{\small

\begin{brule}
\label{bran5}
{\rm
Let $v$ be a vertex. 
If {\bf B}\ref{bran1} applies in ${\cal R}{(G\setminus\{v\})}$ or there exists a vertex
$w$ in ${\cal R}(G\setminus\{v\})$ on which either 
{\bf B}\ref{bran2} or {\bf B}\ref{bran3} applies then branch on $v$. }
\end{brule}
}
\end{boxedminipage}
\end{center}
\caption{\bf \em The branching step yielding a $(1,3/2,3/2)$ drop.}
\label{fig:algoimprovedb5}
\end{figure}

\subsection{The Final branching step}

Finally, if the preprocessing and the branching rules presented thus far do not apply, then note that we are left with a 3-regular graph. In this case, we simply pick a vertex $v$ and branch. 
However, we execute the branching step more carefully in order to simplify the analysis of the drop. 
More precisely, we execute the following step at the end.

\begin{figure}[H]
\begin{center}
\begin{boxedminipage}{.9\textwidth}
{\small
\begin{brule}
\label{bran6}
{\rm
Pick an arbitrary degree $3$ vertex $v$ in $G$ and  let $x$, $y$ and $z$ be the neighbors of $v$.  Then in one branch 
add $v$ into the vertex cover, decrease $k$ by $1$, and 
delete $v$ from the graph. The other branch that excludes $v$ from the 
vertex cover, is performed as follows. 
Delete $x$ from the graph, decrease $k$ by $1$, and obtain ${\cal R}(G\setminus \{x\})$. 
During the process of obtaining ${\cal R}(G\setminus\{x\})$, preprocessing rule~\ref{red:struction} would have been applied on vertices $y$ and $z$ to obtain a `merged' vertex $v_{yz}$ (see proof of correctness of this rule).
Now delete $v_{yz}$ from the graph  ${\cal R}(G\setminus \{x\})$,
 and decrease $k$ by $1$.} 

\end{brule}
}
\end{boxedminipage}
\end{center}

\caption{\bf \em Outline of the last step.}
\label{fig:algoimproved}
\end{figure}

\subsection{Complete Algorithm and Correctness}
%
%
%

%
%
%
\begin{figure}[t]
\begin{center}
\begin{boxedminipage}{.9\textwidth}
{\small

\begin{description}
 \item[Preprocessing Step.] Apply Preprocessing Rules~\ref{red:NT_reduction}, \ref{red:edgeinneighbor} 
and \ref{red:struction} in this order exhaustively on $G$. 
\item[Connected Components.] 
Apply the algorithm on connected components of $G$ separately.  Furthermore, 
if a connected component has size at most 10, then solve the problem optimally in $O(1)$ time. 
\end{description}

\noindent 
{\bf Branching Rules.}\\
These branching rules are applied in this order.\\ 

{\bf B}\ref{bran1}
{\rm
If there is a set $S$ such that \textbf{surplus}$(S)$=\textbf{surplus}$(G)$ and $|S|\geq 2$, then branch on $S$. }\\

{\bf B}\ref{bran2}
{\rm
Let $v$ be a vertex such that $G[N(v)\setminus \{u\}]$ is a clique for some vertex $u$ in $N(v)$. Then in one branch 
add $N(v)$ into the vertex cover, decrease $k$ by $\vert N(v)\vert$, and 
delete $N[v]$ from the graph. In the other branch add $N(u)$ into the vertex cover, decrease $k$ by 
$\vert N(u)\vert$, and delete $N[u]$ from the graph.}\\

{\bf B}\ref{bran3}
{\rm
Let $v$ be a vertex. 
If Preprocessing Rule~\ref{red:edgeinneighbor} can be applied to obtain ${\cal R}(G\setminus \{v\})$ from $G\setminus \{v\}$, then 
branch on $v$. }\\

{\bf B}\ref{bran4}
{\rm
If there exists a vertex $v$ of degree at least $4$ then branch on $v$. }\\
 
{\bf B}\ref{bran5}
{\rm
Let $v$ be a vertex. 
If {\bf B}\ref{bran1} applies in 
${\cal R}(G\setminus\{v\})$ or if there
exists a vertex $w$ in ${\cal R}(G\setminus\{v\})$ on which
{\bf B}\ref{bran2} or {\bf B}\ref{bran3} applies then branch on $v$. }\\

{\bf B}\ref{bran6}
{\rm
Pick an arbitrary degree $3$ vertex $v$ in $G$ and  let $x$, $y$ and $z$ be the neighbors of $v$.  Then in one branch 
add $v$ into the vertex cover, decrease $k$ by $1$, and 
delete $v$ from the graph. The other branch, that excludes $v$ from the 
vertex cover, is performed as follows. 
Delete $x$ from the graph, decrease $k$ by $1$, and obtain ${\cal R}(G\setminus \{x\})$. 
 Now, delete $v_{yz}$ from the graph  ${\cal R}(G\setminus \{x\})$, the vertex that has been created 
by the application of Preprocessing Rule~\ref{red:struction} on $v$ while obtaining ${\cal R}(G\setminus \{x\})$
 and decrease $k$ by $1$.} 

}
\end{boxedminipage}
\end{center}
\caption{\bf \em Outline of the Complete algorithm.}
\label{fig:algoimproved}
\end{figure}

A detailed outline of the algorithm is given in Figure \ref{fig:algoimproved}. 
Note that we have already argued the correctness and analyzed the drops of all steps except the last step, {\bf B}\ref{bran6}.


\noindent
The correctness of this branching rule will follow from the fact that  ${\cal R}(G\setminus \{x\})$ is obtained by applying  Preprocesssing Rule~\ref{red:struction} alone and that too only on the neighbors of $x$, that is, on the degree $2$ 
vertices of $G\setminus \{x\}$ (Lemma~\ref{lem:mainlemma}).  Lemma~\ref{lem:degree3branching} (to appear later) 
shows the correctness of deleting $v_{yz}$ from the graph  ${\cal R}(G\setminus \{x\})$ without branching. 
Thus, the correctness of this algorithm follows from the soundness of the preprocessing rules and the fact that the branching is exhaustive.

\begin{figure}

\begin{center}
    \begin{tabular}{ccccccc}
    \hline
    Rule & {\bf B}\ref{bran1} & {\bf B}\ref{bran2} & {\bf B}\ref{bran3} & {\bf B}\ref{bran4}  & {\bf B}\ref{bran5} & {\bf B}\ref{bran6}\\ \hline \\
    Branching Vector & (1,1) &(1,1)& (1,1) &($\frac{1}{2},\frac{3}{2}$) & ($\frac{3}{2},\frac{3}{2},1$) & ($\frac{3}{2},\frac{3}{2},\frac{5}{2},\frac{5}{2},2$)   \\  \\\hline
    Running time & $2^\mu$ & $2^\mu$ & $2^\mu$ &  $2.1479^\mu$ &  $2.3146^\mu$ & $2.3146^\mu$  \\
    \hline
    \end{tabular}
\end{center}
\caption{A table giving the decrease in the measure due to each branching rule.}
\label{fig:droptable}
\end{figure}

The running time will be dominated by the way {\bf B}\ref{bran6} and the subsequent branching apply. 
We will see that {\bf B}\ref{bran6} is our most expensive branching rule. In fact, 
this step dominates the runtime of the algorithm of $O^*(2.3146^{\mu(G,k)})$ 
due to a branching vector of $(3/2, 3/2, 5/2, 5/2, 2)$. We will argue that when we apply {\bf B}\ref{bran6} on a vertex, say $v$, then 
on either side of the branch we will be able to branch using rules {\bf B}\ref{bran1}, or {\bf B}\ref{bran2}, or {\bf B}\ref{bran3} or 
{\bf B}\ref{bran4}. More precisely, we show that in the branch where 
we include $v$ in the vertex cover, 
\begin{itemize}
\item there is a vertex of degree $4$ in ${\cal R}(G\setminus \{v\})$. Thus,  
{\bf B}\ref{bran4} will apply on the graph ${\cal R}(G\setminus \{v\})$ ( if any of the earlier branching rules applied in this graph, then rule {\bf B}\ref{bran5} would have applied on $G$). 
\item ${\cal R}(G\setminus \{v\})$ has a degree $4$ 
vertex $w$ such that there is a vertex of degree $4$ in the graph ${\cal R}({\cal R}(G\setminus \{v\})\setminus \{w\})$ and thus one of the branching rules  {\bf B}\ref{bran1}, {\bf B}\ref{bran2}, {\bf B}\ref{bran3} or 
{\bf B}\ref{bran4} applies on the graph ${\cal R}({\cal R}(G\setminus \{v\})\setminus \{w\})$. 

\end{itemize} 
Similarly, in the branch where we exclude the vertex $v$ from the solution (and add the vertices $x$ and $v_{yz} $ into the vertex cover), 
we will show that a degree $4$ vertex remains in the reduced graph. This yields the claimed branching vector (see Figure~\ref{fig:droptable}). The rest of the section is geared towards showing this.

\noindent
We start with the following definition. 
\begin{define}
 We say that a graph $G$ is {\em irreducible} if Preprocessing Rules~\ref{red:NT_reduction},  \ref{red:edgeinneighbor} 
and \ref{red:struction} and the branching rules {\bf B}\ref{bran1}, {\bf B}\ref{bran2}, {\bf B}\ref{bran3}, {\bf B}\ref{bran4} 
and {\bf B}\ref{bran5} do not apply on $G$. 
\end{define}
\noindent
Observe that when we apply {\bf B}\ref{bran6}, the current graph is $3$-regular. Thus, after we delete a vertex $v$ from the graph $G$ and apply Preprocessing Rule~\ref{red:struction} we will get a degree $4$ vertex. Our goal is to identify conditions that ensure that the degree $4$ vertices we 
obtain by applying Preprocessing Rule~\ref{red:struction} survive in the graph ${\cal R}(G\setminus \{v\})$. We prove the existence of degree $4$ 
vertices in subsequent branches after applying {\bf B}\ref{bran6} as follows. 
\begin{itemize}
\item We do a closer study of the way
Preprocessing Rules~\ref{red:NT_reduction}, \ref{red:edgeinneighbor} and \ref{red:struction} apply 
on $G\setminus \{v\}$ if Preprocessing Rules~\ref{red:NT_reduction},  \ref{red:edgeinneighbor} 
and \ref{red:struction} and the branching rules {\bf B}\ref{bran1}, {\bf B}\ref{bran2} and {\bf B}\ref{bran3} 
do not apply on $G$.  Based on our observations, we prove some structural properties of the graph ${\cal R}(G\setminus \{v\})$, 
 This is achieved by Lemma~\ref{lem:mainlemma}. 

\item Next, we show that Lemma~\ref{lem:mainlemma}, along with the fact that the graph is irreducible implies a lower bound of 7 on the length of the shortest cycle in the graph (Lemma~\ref{lem:girthlemma}). This lemma allows us to argue that when the 
preprocessing rules are applied, their effect is local. 
\item Finally, Lemmas~\ref{lem:mainlemma} and \ref{lem:girthlemma} together ensure the presence of the required number of degree $4$ vertices in the subsequent branching 
(Lemma~\ref{lem:leftside}).  
\end{itemize}

\subsubsection{Main Structural Lemmas: Lemmas~\ref{lem:mainlemma} and \ref{lem:girthlemma}}

We start with 
some simple well known observations that we use repeatedly in this section. These observations 
follow from results in~\cite{NT2}. We give proofs for completeness.

\begin{lemma}
\label{lem:ntequiv}
Let $G$ be an undirected graph, then the following are equivalent.
\begin{enumerate}[(1)]
\setlength{\itemsep}{-2pt}
\item  Preprocessing Rule~\ref{red:NT_reduction} applies (i.e. All $\frac{1}{2}$ is not the unique solution to the LPVC($G$).)
\item There exists an independent set $I$ of $G$ such that ${\bf surplus}(I)\leq  0$. 
\item There exists an optimal solution $x$ to LPVC($G$) that assigns $0$ to some vertex. 
\end{enumerate}
\end{lemma}
\begin{proof}
\noindent
$(1) \implies (3)$: As we know that the optimum solution is half-integral, there exists an optimum solution that assigns $0$ or $1$ to some vertex. Suppose no vertex is assigned 0. Then, for any vertex which is assigned 1, its value can be reduced to $\frac{1}{2}$ maintaining feasibility (as all its neighbors have been assigned value $\geq \frac{1}{2}$) which is a 
contradiction to the optimality of the given solution.

\noindent
$(3) \implies (2)$: Let $I = V^x_0$, and suppose that ${\bf surplus}(I) > 0$. Then consider the solution $x'$ that assigns $1/2$ to vertices in $I \cup N(I)$, retaining the value of $x$ for the other vertices.
Then $x'$ is a feasible solution whose
objective value $w(x')$ drops from $w(x)$ by $(|N(I)|-|I|)/2 = {\bf surplus}(I)/2 > 0$ which is a contradiction to the optimality of $x$.

\noindent
$(2) \implies (1)$: Setting all vertices in $I$ to $0$, all vertices in $N(I)$ to $1$ and setting the remaining vertices to $\frac{1}{2}$ gives a feasible solution whose objective value is at most $|V|/2$,
and hence all $\frac{1}{2}$ is not the unique solution to LPVC($G$).
\end{proof}

\begin{lemma}
\label{lem:structionequiv}
Let $G$ be an undirected graph, then the following are equivalent.
\begin{enumerate}[(1)]
\setlength{\itemsep}{-2pt}
 \item  Preprocessing Rule~\ref{red:NT_reduction} or \ref{red:edgeinneighbor} or \ref{red:struction} applies. 
\item There exists an independent set $I$ such that ${\bf surplus}(I)\leq  1$. 
\item There exists a vertex $v$ such that an optimal solution $x$ to LPVC($G\setminus \{v\}$) assigns $0$ to some vertex.
\end{enumerate}

\end{lemma}
\begin{proof}
The fact that $(1)$ and $(2)$ are equivalent follows from the definition of the preprocessing rules and Lemma \ref{lem:ntequiv}.

\noindent
$(3) \implies (2)$. By Lemma \ref{lem:ntequiv}, there exists an independent set $I$ in $G\setminus \{v\}$ whose surplus is at most $0$. The same set will have surplus at most $1$ in $G$.

\noindent
$(2) \implies (3)$. Let $v \in N(I)$. 
Then $I$ is an independent set in $G~\setminus\{v\}$ with surplus at most $0$, and hence by Lemma \ref{lem:ntequiv}, there exists an optimal solution to 
LPVC($G\setminus \{ v \}$) that assigns $0$ to some vertex.
\end{proof}

\noindent
We now prove an auxiliary lemma about the application of Preprocessing Rule~\ref{red:struction} which will be useful in simplifying later proofs. 

\begin{lemma}
\label{lem:char1}
Let $G$ be a graph and $G_R$ be the graph obtained from $G$ by applying Preprocessing Rule~\ref{red:struction}  
on an independent set $Z$. Let $z$ denote the newly added vertex corresponding to $Z$ in $G_R$.  
\begin{enumerate}
\item If $G_R$ has an independent set $I$ such that ${\bf surplus}(I)= p$, then $G$ also has an independent set $I'$ such 
that ${\bf surplus}(I')= p$ and $|I'|\geq |I|$. 
\item Furthermore, if $z\in I\cup N(I)$ then $|I'|>|I|$. 
\end{enumerate}
\end{lemma}
\begin{proof}
Let $Z$ denote the minimum surplus independent set on  which 
  Preprocessing Rule~\ref{red:struction}  has been applied and $z$ denote 
the newly added vertex.  Observe that since Preprocessing Rule~\ref{red:struction} applies on $Z$, 
we have that $Z$ and $N(Z)$ are independent sets, $|N(Z)|=|Z|+1$ and $|N(Z)|\geq 2$.  

Let $I$ be an independent set of $G_R$ such that  ${\bf surplus}(I)= p$. 
\begin{itemize}
\item If both $I$ and $N(I)$ do 
not contain $z$ then we have that $G$ has an independent set $I$ such that  ${\bf surplus}(I) =p$. 
\item Suppose $z\in I$. Then consider the following set:  $I':= I\setminus \{z\}\cup N(Z)$.  Notice that $z$ represents $N(Z)$ 
and thus $I$ do not have any neighbors of $N(Z)$. This implies that $I'$ is an independent set in $G$. 
Now we will show that ${\bf surplus}(I') =p$. We know that $|N(Z)|=|Z|+1$ and 
$N(I')=N(I)\cup Z$.  Thus,
\begin{eqnarray*}
 |N(I')| - |I'| & = & (|N(I)|+|Z|)- |I'| \\
        &=& (|N(I)|+|Z|)-(|I|-1+|N(Z)|)\\
        &=& (|N(I)|+|Z|)-(|I|+|Z|)\\
        &=& |N(I)|-|I| = {\bf surplus}(I) =p.
  \end{eqnarray*}

\item Suppose $z \in N(I)$. Then consider the following set:  $I':= I \cup Z$.  Notice that $z$ represents $N(Z)$ and 
 since $z\notin I$ we have that $I$ do not have any neighbors of $Z$. 
This implies that $I'$ is an independent set in $G$.  
We  show that ${\bf surplus}(I') =p$. We know that $|N(Z)|=|Z|+1$. Thus,
\begin{eqnarray*}
 |N(I')| - |I'| & = & (|N(I)|-1+|N(Z)|)- |I'| \\
        &=& (|N(I)|-1+|N(Z)|)-(|I|+|Z|)\\
        &=& (|N(I)|+|Z|)-(|I|+|Z|)\\
        &=& |N(I)|-|I| = {\bf surplus}(I) =p.
  \end{eqnarray*}
\end{itemize}
From the construction of $I'$, it is clear that $|I'|\geq |I|$ and if $z\in (I\cup N(I))$ then $|I'|>|I|$.  
This completes the proof.
\end{proof}

\noindent
We now give some definitions that will be useful in formulating the statement of the main structural lemma. 
\begin{define}
 Let $G$ be a graph and ${\cal P}=P_1,P_2, \cdots , P_\ell$ be a sequence of exhaustive applications of 
Preprocessing Rules~\ref{red:NT_reduction},  \ref{red:edgeinneighbor} and \ref{red:struction} applied in this order 
on $G$ to obtain $G'$.  Let ${\cal P}_3=P_a,P_b, \cdots , P_t$ be the subsequence of $\cal P$ restricted to Preprocessing Rule~\ref{red:struction}. 
Furthermore let $Z_j$, $j\in \{a,\ldots,t\}$ denote the minimum surplus independent set corresponding to $P_t$ on which 
the  Preprocessing Rule~\ref{red:struction}  has been applied and $z_j$ denote the newly added vertex (See Lemma~\ref{lem:struction}). 
Let $Z^*=\{z_j~|~j\in \{a,\ldots,t\}\}$ be the set of these newly added vertices. 
\begin{itemize}
 \item We say that an applications of Preprocessing Rule~\ref{red:struction} is trivial if 
the minimum surplus independent set $Z_j$ on which $P_j$ is applied has size $1$, that is, $|Z_j|=1$. 
\item We say that all applications of  Preprocessing Rule~\ref{red:struction} are independent if for all $j\in \{a,\ldots,t\}$, 
$N[Z_j]\cap Z^*=\emptyset$. 
\end{itemize}
\end{define}
\noindent 
Essentially, independent applications of Preprocessing Rule~\ref{red:struction} mean that the set on which the rule is applied, and all its 
neighbors are vertices in the original graph.

\noindent
Next, we state and prove one of the main structural lemmas of this section.

\begin{lemma}
\label{lem:mainlemma}
 Let $G=(V,E)$ be a graph on which Preprocessing Rules~\ref{red:NT_reduction},  \ref{red:edgeinneighbor} 
and \ref{red:struction} and the branching rules {\bf B}\ref{bran1}, {\bf B}\ref{bran2} and {\bf B}\ref{bran3} 
do not apply.  
Then for any vertex $v\in V$, 
\begin{enumerate}
\setlength{\itemsep}{-2pt}
 \item  preprocessing Rules~\ref{red:NT_reduction} and  \ref{red:edgeinneighbor} have not been applied while obtaining ${\cal R}(G\setminus \{v\})$ from $G\setminus \{v\}$; 
\item  and all applications of the Preprocessing Rule~\ref{red:struction} while obtaining ${\cal R}(G\setminus \{v\})$ from $G\setminus \{v\}$ are independent and trivial. 
\end{enumerate}
\end{lemma}
\begin{proof} Fix a vertex $v$.  Let  $G_0=G\setminus \{v\}, G_1,\ldots ,G_t={\cal R}(G\setminus \{v\})$ be a sequence of 
graphs obtained  by applying  Preprocessing Rules~\ref{red:NT_reduction},  \ref{red:edgeinneighbor} 
and \ref{red:struction} in this order to obtain the reduced graph ${\cal R}(G\setminus \{v\})$. 

We first observe that  Preprocessing Rule  \ref{red:edgeinneighbor} never 
applies in obtaining ${\cal R}(G\setminus \{v\})$ from $G\setminus \{v\}$ since otherwise,  {\bf B}\ref{bran3} 
would have applied on $G$.  Next, we show that Preprocessing Rule~\ref{red:NT_reduction} 
does not apply.  Let $q$ be the least integer such that  Preprocessing Rule~\ref{red:NT_reduction}  applies on 
$G_q$ and it does not apply to any graph $G_{q'}$, $q'<q$.  Suppose that $q\geq 1$. 
Then, only Preprocessing Rule~\ref{red:struction} has been applied on $G_0,\ldots,G_{q-1}$. 
This implies that $G_q$ has an independent set $I_q$ such that ${\bf surplus}(I_q) \leq 0$. Then, by
Lemma~\ref{lem:char1}, $G_{q-1}$ also has an independent set $I_q'$ such that ${\bf surplus}(I_q') \leq 0$ 
and thus by Lemma~\ref{lem:ntequiv} Preprocessing Rule~\ref{red:NT_reduction} applies to $G_{q-1}$.  This contradicts  
the assumption that on $G_{q-1}$ Preprocessing Rule~\ref{red:NT_reduction} does not apply.  Thus, we conclude that 
$q$ must be zero. So, $G\setminus \{v\}$ has an independent set $I_0$ such that ${\bf surplus}(I_0) \leq 0$ in $G\setminus \{v\}$ 
and thus $I_0$ is an independent set in $G$ such that ${\bf surplus}(I_0) \leq 1$ in $G$. By Lemma~\ref{lem:structionequiv} this implies that either of 
Preprocessing Rules~\ref{red:NT_reduction},  \ref{red:edgeinneighbor}  or  \ref{red:struction} is applicable on $G$, a contradiction to 
the given assumption.

 Now we show the second part of the lemma.  By the first part we know that the $G_i$'s have been obtained by applications of 
Preprocessing Rule~\ref{red:struction} alone.  Let $Z_i$, $0\leq i \leq t-1$ be the sets  in $G_i$ on which 
Preprocessing Rule~\ref{red:struction} has been applied. Let  the newly added vertex corresponding to $N(Z_i)$ in this process be  
$z_i'$. We now make the following claim.
\begin{claim}
\label{imp:claim}
For any $i\geq 0$, if $G_i$ has an independent set $I_i$ such that ${\bf surplus}(I_i)= 1$, then $G$ has an independent set $I$ such that $|I| \geq |I_i|$ and ${\bf surplus}(I) = 2$. Furthermore, if $(I_i \cup N(I_i)) \cap \{z_1,\dots,z_{i-1}\} \neq \phi$, then $|I| >|I_i|$.
\end{claim}
\begin{proof}
We prove the claim by induction on the length of the sequence of graphs. 
For the base case consider $q=0$. Since Preprocessing Rules~\ref{red:NT_reduction},~\ref{red:edgeinneighbor}, and ~\ref{red:struction} 
do not apply on $G$, we have that ${\bf surplus}(G)\geq 2$.  Since $I_0$ is an independent set in $G\setminus \{v\}$ we have that 
$I_0$ is an independent set in $G$ also. Furthermore since ${\bf surplus}(I_0)=1$ in $G\setminus \{v\}$, we have that 
${\bf surplus}(I_0)=2$ in $G$, as ${\bf surplus}(G)\geq 2$.  This implies that $G$ has an independent set  $I_0$  with 
${\bf surplus}(I_0)=2={\bf surplus}(G)$. Furthermore, since $G_0$ does not have any newly introduced vertices, 
the last assertion is vacuously true. Now let $q\geq 1$. 
Suppose that $G_q$ has  a set $|I_q|$ and ${\bf surplus}(I_q) =1$. Thus, by Lemma~\ref{lem:char1}, $G_{q-1}$ also has an independent set $I_q'$ such that  $|I_q'|\geq |I_q|$ and ${\bf surplus}(I_q')=1$ . Now by the induction hypothesis, $G$  has an independent set $I$ such that $|I|\geq |I_q'|\geq |I_q|$ and ${\bf surplus}(I) =2={\bf surplus}(G)$.  

Next we consider the case when $(I_q\cup N(I_q)) \cap \{z_1',\ldots, z_{q-1}' \}\neq \emptyset$. If $z_{q-1}'\notin I_q\cup N(I_q)$ then we have 
that $I_{q}$ is an independent set in $G_{q-1}$ such that $(I_q\cup N(I_q)) \cap \{z_1',\ldots, z_{q-2}' \}\neq \emptyset$. Thus, by induction we have 
that  $G$ has  an independent set $I$ such that $|I|>|I_q|$ and ${\bf surplus}(I) =2={\bf surplus}(G)$.  
  On the other hand, if  
$z_{q-1}'\in I_q\cup N(I_q)$ then by Lemma~\ref{lem:char1},  we know 
that $G_{q-1}$ has an independent set $I_q'$ such that  $|I_q'|>|I_q|$ and ${\bf surplus}(I_q')=1$ . Now by induction hypothesis we know 
that $G$ has  an independent set $I$ such that $|I|\geq |I_q'|>|I_q|$ and ${\bf surplus}(I) =2={\bf surplus}(G)$.  
This concludes the proof of the claim.  
\end{proof}

We first show that all the applications of  Preprocessing Rule~\ref{red:struction} are trivial.   Claim~\ref{imp:claim} implies that if we have a non-trivial 
application of Preprocessing Rule~\ref{red:struction}  then it implies that $G$ has an independent set $I$ such that 
$|I|\geq 2$ and ${\bf surplus}(I) =2={\bf surplus}(G)$.  Then,  {\bf B}\ref{bran1}  would apply on $G$, a contradiction. 

Finally, we show that all the applications of Preprocessing Rule~\ref{red:struction} are  independent. Let $q$ be the least integer such that  
the application of Preprocessing Rule~\ref{red:struction} on $G_q$ is not independent. That is, the application of Preprocessing 
Rule~\ref{red:struction} on $G_{q'}$, $q'<q$, is trivial and independent. Observe that $q\geq 1$. We already know that every application of Preprocessing 
Rule~\ref{red:struction}  is trivial. This implies that the set $Z_q$ contains a single vertex. Let $Z_q=\{z_q\}$. Since the application of 
Preprocessing Rule~\ref{red:struction} on $Z_q$ is not independent we have that 
$(Z_q \cup N(Z_q))\cap \{z_1', \cdots, z_{q-1}'\}\neq \emptyset$.  We also know that ${\bf surplus}(Z_q) =1$ and thus by Claim~\ref{imp:claim} we have that $G$ has an independent set $I$ such that $|I|\geq 2> |Z_q| $ and 
${\bf surplus}(I) =2={\bf surplus}(G)$. This implies that  {\bf B}{\ref{bran1}}  would apply on $G$, a contradiction. 
Hence, we conclude that all the applications of  Preprocessing Rule~\ref{red:struction} are independent. This proves the lemma.  
\end{proof}
\noindent

\noindent
Let $g(G)$ denote the girth of the graph, that is, the length of the smallest cycle in $G$. 
Our next goal of this section is to obtain a lower bound on the girth of an irreducible graph. 
Towards this, we first introduce the notion of an {\em untouched} vertex. 
\begin{define}
We say that a vertex $v$ is \emph{untouched} by an application of Preprocessing Rule \ref{red:edgeinneighbor} 
or Preprocessing Rule \ref{red:struction}, if $\{v\}\cap (Z\cup N(Z))=\phi$, where $Z$ is the set on which the rule is applied.
\end{define}
\noindent

\noindent 
We now prove an auxiliary lemma regarding the application of the preprocessing rules on 
graphs of a certain girth and following that, we will prove a lower bound on the girth of irreducible graphs.
\begin{lemma}
\label{lem:helpme} 
Let $G$ be a graph on which Preprocessing Rules~\ref{red:NT_reduction},  \ref{red:edgeinneighbor} 
and \ref{red:struction} and the branching rules {\bf B}\ref{bran1}, {\bf B}\ref{bran2}, {\bf B}\ref{bran3} do not apply and suppose that $g(G)\geq 5$. Then for any vertex $v\in V$, any vertex $x \notin N_2[v]$ is untouched by the preprocessing rules applied to obtain the graph ${\cal R}(G\setminus \{v\})$ from $G\setminus\{v\}$ and has the same degree as it does in $G$. 
\end{lemma}

\begin{proof}
Since the preprocessing rules do not apply in $G$, the minimum degree of $G$ is at least $3$ and since the graph $G$ does not have cycles of length 3 or 4, for any vertex $v$,  the neighbors of $v$ are independent and there are no edges between vertices in the first and second neighborhood of $v$.

 We know by Lemma~\ref{lem:mainlemma} that only Preprocessing Rule~\ref{red:struction} applies on the graph $G\setminus \{v\}$ and it applies only in a trivial and independent way. Let $U=\{u_1,\dots,u_t\}$ be the degree 3 neighbors of $v$ in $G$ and let $D$ represent the set of the remaining (high degree) neighbors of $v$. Let $P_1,\dots, P_l$ be the sequence of applications of Preprocessing Rule \ref{red:struction} on the graph $G\setminus \{v\}$, let $Z_i$ be the minimum surplus set corresponding to the application of $P_i$ and let $z_i$ be the new vertex created during the application of $P_i$.

We prove by induction on $i$, that \begin{itemize} \item the application $P_i$ corresponds to a vertex $u_j\in U$, \item any vertex $x\notin N_2[v]\setminus D$ is untouched by this application, and \item after the application of $P_i$,  the degree of $x\notin N_2[v]$ in the resulting graph is the same as that in $G$.\end{itemize}

In the base case, $i=1$. Clearly, the only vertices of degree 2 in the graph $G\setminus \{v\}$ are the degree 3 neighbors of $v$. Hence, the application $P_1$ corresponds to some $u_j\in U$. Since the graph $G$ has girth at least 5, no vertex in $D$ can lie in the set $\{u_j\}\cup N(u_j)$ and hence must be untouched by the application of $P_1$. Since $u_j$ is a neighbor of $v$, it is clear that the application of $P_1$ leaves any vertex disjoint from $N_2[v]$ untouched. Now, suppose that after the application of $P_1$, a vertex $w$ disjoint from  $N_2[v]\setminus D$ has lost a degree. Then, it must be the case that the application of $P_1$ identified two of $w$'s neighbors, say $w_1$ and $w_2$ as the vertex $z_1$. But since $P_1$ is applied on the vertex $u_j$, this implies the existence of a 4 cycle $u_j, w_1,w,w_2$ in $G$, which is a contradiction.

We assume as induction hypothesis that the claim holds for all $i^\prime$ such that $1\leq i^\prime<i$ for some $i>1$. Now, consider the application of $P_i$. By Lemma~\ref{lem:mainlemma}, this application cannot be on any of the vertices created by the application of $P_l$ ($l<i$), and by the induction hypothesis, after the application of $P_{i-1}$, any vertex disjoint from $N_2[v]\setminus D$ remains untouched and retains the degree (which is $\geq 3$) it had in the original graph. Hence, the application of $P_i$ must occur on some vertex $u_j\in U$. Now, suppose that a vertex $w$ disjoint from $N_2[v]\setminus D$ has lost a degree. Then, it must be the case that $P_i$ identified two of $w$'s neighbors say $w_1$ and $w_2$ as the vertex $z_i$. Since $P_i$ is applied on the vertex $u_j$, this implies the existence of a 4 cycle $u_j, w_1,w,w_2$ in $G$, which is a contradiction. Finally, after the application of $P_i$, since no vertex outside $N_2[v]\setminus D$ has ever lost degree and they all had degree at least 3 to begin with, we cannot apply Preprocessing Rule~\ref{red:struction} any further. This completes the proof of the claim.

Hence, after applying Preprocessing Rule~\ref{red:struction} exhaustively on $G\setminus \{v\}$, any vertex disjoint from $N_2[v]$ is untouched and has the same degree as in the graph $G$. This completes the proof of the lemma.
\end{proof}

\noindent
Recall that the graph is irreducible if none of the preprocessing rules or branching rules {\bf B}\ref{bran1} through {\bf B}\ref{bran5} apply, i.e: the algorithm has reached {\bf B}\ref{bran6}.

\begin{figure}[t]
\begin{minipage}[b]{0.5\linewidth}
\centering
\includegraphics[height=110 pt, width=140 pt]{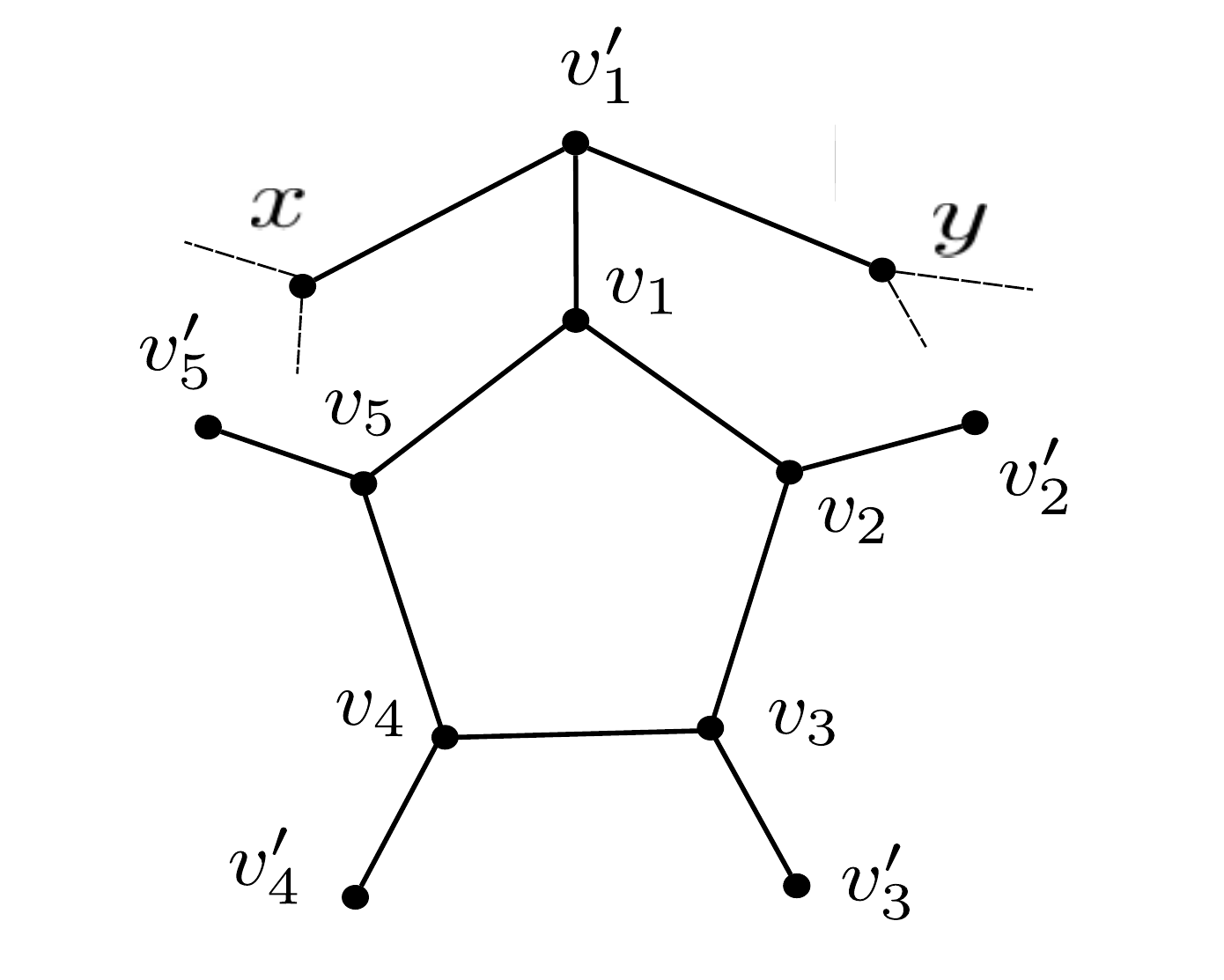}

\includegraphics[height=110 pt, width=140 pt]{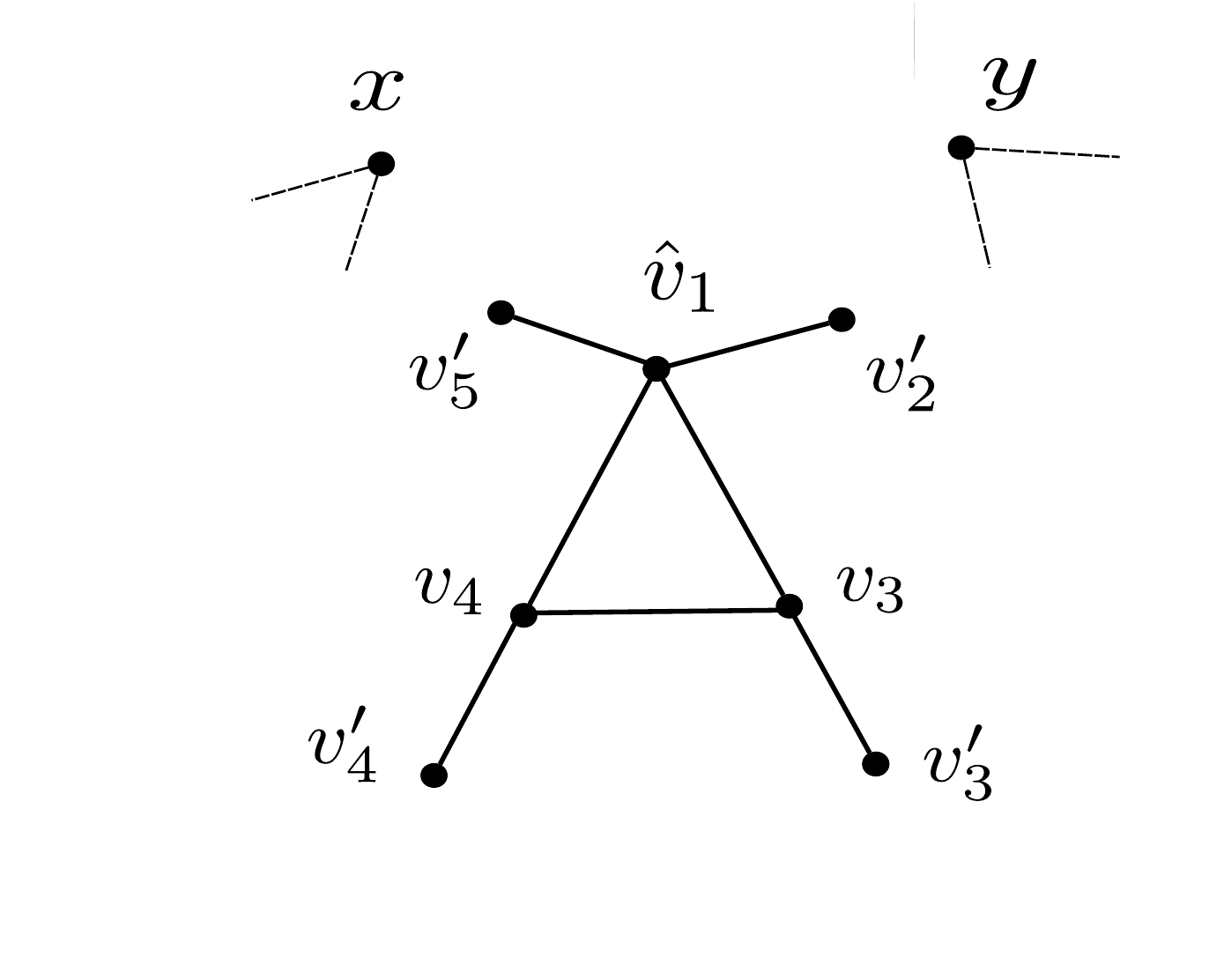}
\end{minipage}
\hspace{0.5cm}
\begin{minipage}[b]{0.5\linewidth}
\centering
\includegraphics[height=100 pt, width=130 pt]{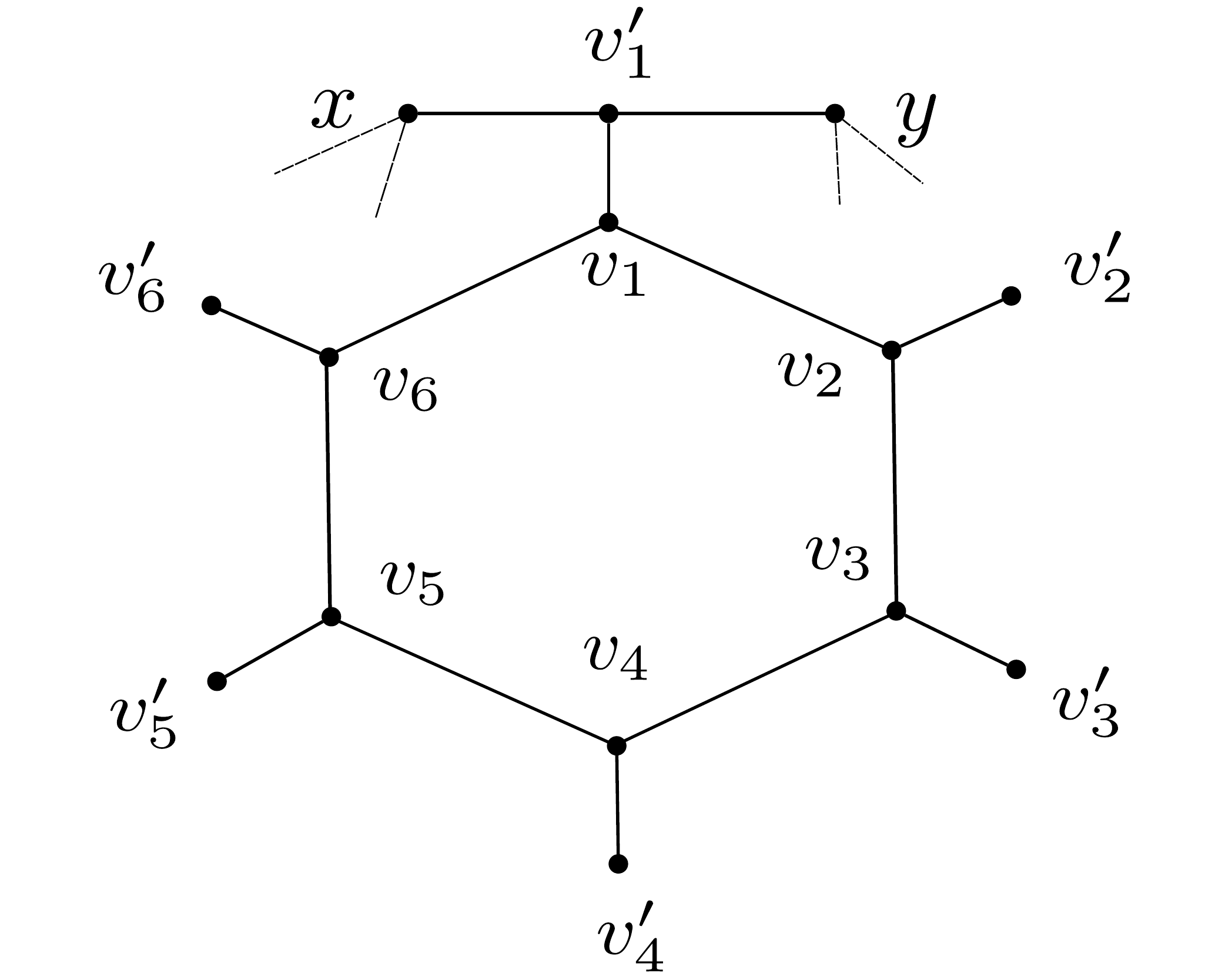}

\includegraphics[height=100 pt, width=130 pt]{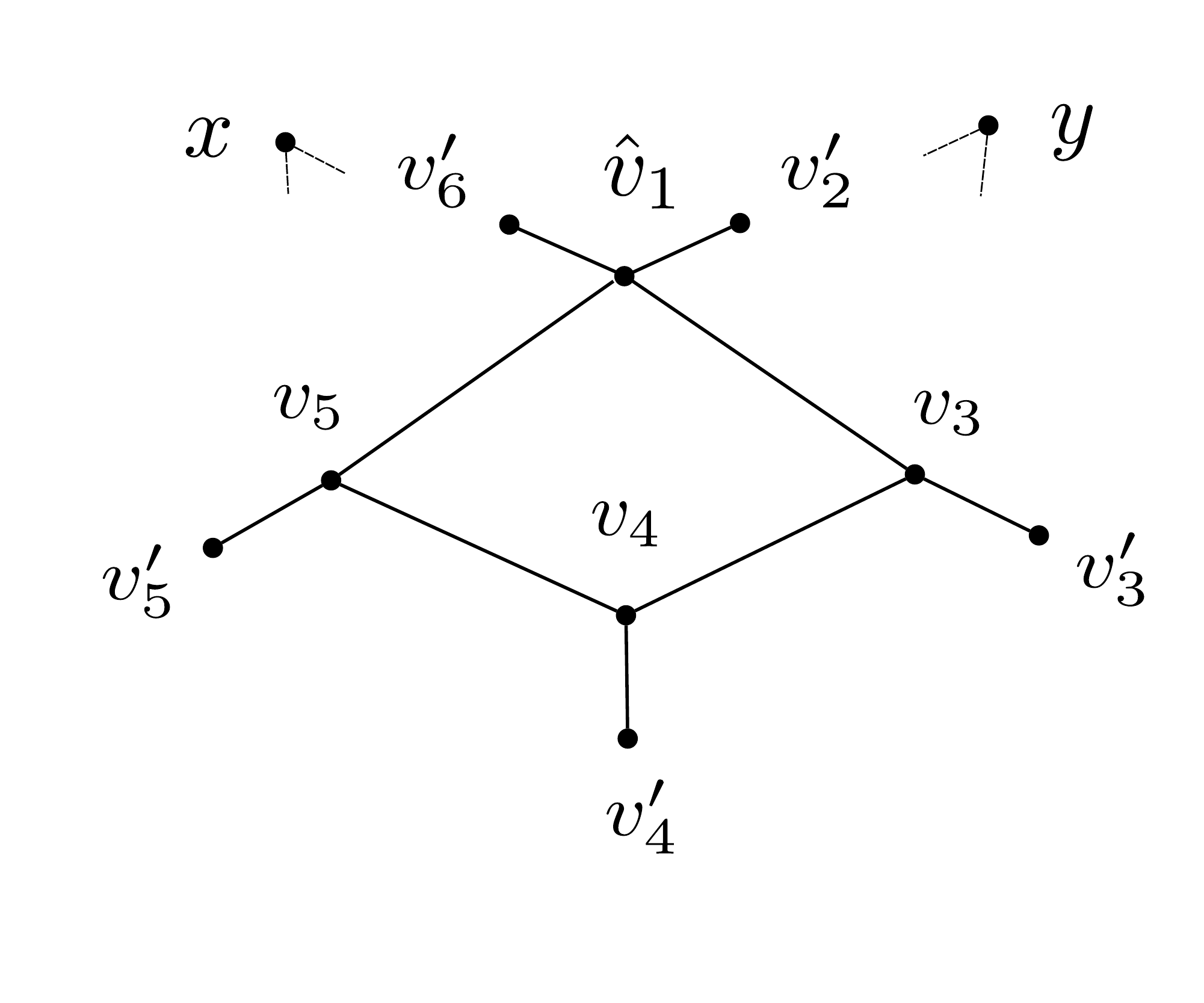}
\end{minipage}

\caption{Cases of Lemma~\ref{lem:girthlemma} when there is a 5 cycle or a 6 cycle in the graph}
\label{fig:c5c6}
\end{figure}

\begin{lemma}
\label{lem:girthlemma}
Let $G$ be a connected $3$-regular irreducible graph with at least $11$ vertices. Then, $g(G)\geq 7$. 
\end{lemma}
\begin{proof}
\begin{enumerate}
 \item Suppose $G$ contains a triangle $v_1,v_2,v_3$. Let $v_4$ be the remaining neighbor of $v_1$. Now, 
$G[N(v_1)\setminus \{v_4\}]$ is a clique,
 which implies that branching rule {\bf B}\ref{bran2} applies and hence contradicts the irreducibilty of $G$. Hence, $g(G)\geq 4$.

\item Suppose $G$ contains a cycle $v_1,v_2,v_3,v_4$ of length 4. Since $G$ does not contain triangles, it must be the case that $v_1$ and $v_3$ are independent. Recall that $G$ has minimum surplus 2, and hence surplus of the set $\{v_1,v_3\}$ is at least $2$. Since $v_2$ and $v_4$ account for two neighbors of both $v_1$ and $v_3$, the neighborhood of $v_1\cup v_3$ can contain at most $2$ more vertices ($G$ is 3 regular). Since the minimum surplus of $G$ is 2, $\vert N(\{v_1,v_2\})\vert=4$ and hence $\{v_1,v_3\}$ is a minimum surplus set of size 2, which implies that branching rule {\bf B}\ref{bran1} applies and hence contradicts the irreduciblity of $G$. Hence, $g(G)\geq 5$.

\item Suppose that $G$ contains a 5 cycle $v_1,\dots,v_5$. Since $g(G)\geq 5$, this cycle does not contain chords. Let $v_i^\prime$ denote the remaining neighbor of the vertex $v_i$ in the graph $G$. Since there are no triangles or 4 cycles, $v_i^\prime\neq v_j^\prime$ for any $i\neq j$, and for any $i$ and $j$ such that $\vert i-j\vert=1$, $v^\prime_i$ and $v_j^\prime$ are independent. Now, we consider the following 2 cases.

\textbf{Case 1:} Suppose that for every $i,j$ such that $\vert i-j\vert\neq 1$, $v_i^\prime$ and $v_j^\prime$ are adjacent. Then, since $G$ is a connected 3-regular graph, $G$ has size 10, which is a contradiction.\\
\textbf{Case 2:} Suppose that for some $i,j$ such that $\vert i-j\vert\neq 1$, $v_i^\prime$ and $v_j^\prime$ are independent (see Figure ~\ref{fig:c5c6}). Assume without loss of generality that $i=1$ and $j=3$. Consider the vertex $v_1^\prime$ and let $x$ and $y$ be the remaining 2 neighbors of $v_1^\prime$ (the first neighbor being $v_1$). 
Note that $x$ or $y$ cannot be incident to $v_3$, since otherwise $x$ or $y$ will coincide with $v_3^\prime$.  Hence, $v_3$ is disjoint from $N_2[v_1^\prime]$.
By Lemma~\ref{lem:mainlemma} and Lemma~\ref{lem:helpme}, only Preprocessing Rule~\ref{red:struction} applies and the applications are only on the vertices $v_1$, $x$ and $y$ and leaves $v_3$ untouched and the degree of vertex $v_3$ unchanged. Now, let $\hat v_1$ be the vertex which is created as a result of applying Preprocessing Rule~\ref{red:struction}\ on $v_1$. Let $\hat v_4$ be the vertex created when $v_4$ is identified with another vertex during some application of Preprocessing Rule~\ref{red:struction}. If $v_4$ is untouched, then we let $\hat v_4=v_4$. Similarly, let  $\hat v_3^\prime$ be the vertex created when  $v_3^\prime$ is identified with another vertex during some application of Preprocessing Rule~\ref{red:struction} . If $v_3^\prime$ is untouched, then we let $\hat v_3^\prime=v_3^\prime$. Since $v_3$ is untouched and its degree remains 3 in the graph ${\cal R}(G\setminus \{v\})$, the neighborhood of $v_3$ in this graph can be covered by a 2 clique $\hat v_1,\hat v_4$ and a vertex $\hat v_3^\prime$, which implies that branching rule {\bf B}\ref{bran2} applies in this graph, implying that branching rule {\bf B}\ref{bran5} applies in the graph $G$, contradicting the irreduciblity of $G$. Hence, $g(G)\geq 6$.

\item Suppose that $G$ contains a 6 cycle $v_1,\dots,v_6$. Since $g(G)\geq 6$, this cycle does not contain chords. Let $v_i^\prime$ denote the remaining neighbor of each vertex $v_i$ in the graph $G$. Let $x$ and $y$ denote the remaining neighbors of $v_1^\prime$ (see Figure~\ref{fig:c5c6}). Note that both $v_3$ and $v_5$ are disjoint from $N_2[v_1^\prime]$ (if this were not the case, then we would have cycles of length $\leq 5$). Hence, by Lemma~\ref{lem:mainlemma} and Lemma~\ref{lem:helpme}, we know that only Preprocessing Rule~\ref{red:struction} applies and the applications are only on the vertices  $v_1$, $x$ and $y$, vertices $v_3$ and $v_5$ are untouched, and the degree of $v_3$ and $v_5$ in the graph ${\cal R}(G\setminus\{ v_1^\prime\})$ is 3. Let $\hat v_1$ be the vertex which is created as a result of applying $P_3$ on $v_1$. Let $\hat v_4$ be the vertex created when $v_4$ is identified with another vertex during some application of $P_3$. If $v_4$ is untouched, then we let $\hat v_4=v_4$. Now, in the graph ${\cal R}(G\setminus \{v_1^\prime\})$, the vertices $v_3$ and $v_5$ are independent and share two neighbors $\hat v_1$ and $\hat v_4$. The fact that they have degree 3 each and the surplus of graph ${\cal R}(G\setminus \{v_1^\prime\} )$ is at least 2 (Lemma~\ref{lem:mainlemma}, Lemma~\ref{lem:structionequiv}) implies that $\{v_3,v_5\}$ is a minimum surplus set of size at least 2 in the graph ${\cal R}(G\setminus\{ v_1^\prime\})$, which implies that branching rule {\bf B}\ref{bran2} applies in this graph, implying that branching rule {\bf B}\ref{bran5} applies in the graph $G$, contradicting the irreduciblity of $G$. Hence, $g(G)\geq 7$.
\end{enumerate}
\noindent
This completes the proof of the lemma. \end{proof}

\subsubsection{Correctness and Analysis of the last step}
\noindent
In this section we combine all the results proved above and show the existence of 
degree $4$ vertices in subsequent branchings after {\bf B}\ref{bran6}. 
Towards this we prove the following lemma. 
\begin{lemma}
\label{lem:leftside}
Let $G$ be a connected $3$ regular irreducible graph on at least 11 vertices. Then, for any vertex $v\in V$,  
\begin{enumerate}
\item ${\cal R}(G\setminus \{v\})$ contains three degree $4$ vertices, say $w_1, w_2, w_3$; and 
\item for any $w_i$, $i\in \{1,2,3\}$, ${\cal R}({\cal R}(G\setminus \{v\})\setminus \{w_i\})$ contains $w_j$, $i\neq j$ as a degree $4$ vertex. 
\end{enumerate}
\end{lemma}

\begin{proof}
\begin{enumerate}
 \item Let $v_1,v_2,v_3$ be the neighbors of $v$. Since $G$ was irreducible, {\bf B}\ref{bran1}, {\bf B}\ref{bran2}, {\bf B}\ref{bran3} do not apply on ${\cal R}(G\setminus \{v\})$ (else {\bf B}\ref{bran5} would have applied on $G$). By Lemma~\ref{lem:mainlemma} and Lemma~\ref{lem:helpme}, we know that only Preprocessing Rule~\ref{red:struction} would have applied and the applications are only on these three vertices . Let $w_1,w_2$ and $w_3$ be the three vertices which are created as a result of applying Preprocessing Rule~\ref{red:struction} on these three vertices respectively. We claim that the degree of each $w_i$ in the resulting graph is 4. Suppose that the degree of $w_j$ is at most 3 for some $j$. But this can happen only if there was an edge between two vertices which are at a distance of 2 from $v$, that is, a path of length 3 between $w_i$ and $w_j$ for some $i\neq j$. This implies the existence of a cycle of length 5 in $G$, which contradicts Lemma~\ref{lem:girthlemma}.
\item Note that, by Lemma~\ref{lem:helpme}, it is sufficient to show that $w_i$ is disjoint from $N_2[w_j]$ for any $i\neq j$. Suppose that this is not the case and let $w_i$ lie in $N_2[w_j]$. First, suppose that $w_i$ lies in $N_2[w_j]\setminus N_1[w_j]$ and there is no $w_k$ in $N_1[w_i]$. Let $x$ be a common neighbor of $w_i$ and $w_j$. This implies that, in $G$, $x$ has paths of  length 3 to $v$ via $w_i$ and via $w_j$, which implies the existence of a cycle of length at most 6, a contradiction. Now, suppose that $w_i$ lies in $N_1[w_j]$. But this can happen only if there was an edge between two vertices which are at a distance of 2 from $v$. This implies the existence of a cycle of length 5 in $G$, contradicting Lemma~\ref{lem:girthlemma}.\end{enumerate}\end{proof}
\noindent
The next lemma shows the correctness of deleting 
$v_{yz}$ from the graph  ${\cal R}(G\setminus \{x\})$ without branching.


\begin{lemma}
\label{lem:degree3branching}Let $G$ be a connected irreducible graph on at least 11 vertices, $v$ be a vertex of degree 3, and $x,y,z$ be the set of its neighbors. Then, $G\setminus\{x\}$ contains a vertex cover of size at most $k$ which excludes $v$ if and only if ${\cal R}(G\setminus \{x\})$ contains a vertex cover of size at most $k-3$ which contains $v_{yz}$, where $v_{yz}$ is the vertex created in the graph $G\setminus \{x\}$ by the application of Preprocessing Rule~\ref{red:struction} on the vertex $v$.
\end{lemma}

\begin{proof}
We know by Lemma~\ref{lem:helpme} that there will be exactly 3 applications of Preprocessing Rule~\ref{red:struction} in the graph $G\setminus \{x\}$, and they will be on the three neighbors of $x$. Let $G_1$, $G_2$, $G_3$ be the graphs which result after each such application, in that order. We assume without loss of generality that the third application of Preprocessing Rule~\ref{red:struction} is on the vertex $v$.

By the correctness of Preprocessing Rule~\ref{red:struction}, if $G\setminus \{x\}$ has a vertex cover of size at most $k$ which excludes $v$, then $G_2$ has a vertex cover of size at most $k-2$ which excludes $v$. Since this vertex cover must then contain $y$ and $z$, it is easy to see that $G_3$ contains a vertex cover of size at most $k-3$ containing $v_{yz}$.

Conversely, if $G_3$ has a vertex cover of size at most $k-3$ containing $v_{yz}$, then replacing $v_{yz}$ with the vertices $y$ and $z$ results in a vertex cover for $G_2$ of size at most $k-2$ containing $y$ and $z$ (by the correctness of Preprocessing Rule~\ref{red:struction}). Again, by the correctness of Preprocessing Rule~\ref{red:struction}, it follows that $G\setminus \{x\}$ contains a vertex cover of size at most $k$ containing $y$ and $z$. Since $v$ is adjacent to only $y$ and $z$ in $G\setminus \{x\}$, we may assume that this vertex cover excludes $v$.
\end{proof}

\noindent
Thus, when branching rule {\bf B}\ref{bran6} applies on the graph $G$, we know the following about the graph. 
\begin{itemize}
\item $G$ is a $3$ regular graph. This follows from the fact that Preprocessing Rules~\ref{red:NT_reduction},~\ref{red:edgeinneighbor} 
and \ref{red:struction} and the branching rule {\bf B}\ref{bran4}  do not apply. 
\item $g(G)\geq 7$. This follows from Lemma~\ref{lem:girthlemma}. 
\end{itemize}
\noindent
Let $v$ be an arbitrary vertex and $x$, $y$ and $z$ be the neighbors of $v$.  Since $G$ is irreducible,  Lemma~\ref{lem:leftside} 
implies that ${\cal R}(G\setminus \{x\})$ contains 3 degree $4$ vertices, $w_1$, $w_2$ and $w_3$. We let $v_{yz}$ be $w_1$. Lemma~\ref{lem:leftside} also implies that for any $i$, the graph ${\cal R}({\cal R}(G\setminus \{x\})\setminus \{w_i\})$ contains 2 degree $4$ vertices. Since the vertex $v_{yz}$ is one of the three degree 4 vertices, in the graph ${\cal R}({\cal R}(G\setminus \{x\})\setminus v_{yz})$, the vertices $w_2$ and $w_3$ have degree 4 and  one of the branching rules 
{\bf B}\ref{bran1}, or  {\bf B}\ref{bran2}, or {\bf B}\ref{bran3} or {\bf B}\ref{bran4}  will apply in this graph. 
Hence, we combine the execution of the rule {\bf B}\ref{bran6} along with the subsequent execution of one of the rules {\bf B}\ref{bran1}, {\bf B}\ref{bran2}, {\bf B}\ref{bran3} or {\bf B}\ref{bran4} (see Fig.~\ref{fig:B5B6}). To analyze the drops in the measure for the combined application of these rules, we consider each root to leaf path in the tree of Fig.~\ref{fig:B5B6}~(b) and argue the drops in each path.

\begin{itemize}
\item Consider the subtree in which $v$ is not picked in the vertex cover from $G$, that is, $x$ is picked in the vertex cover, following which we branch on some vertex
$w$ during the subsequent branching, from the graph ${\cal R}({\cal R}(G\setminus \{x\})\setminus v_{yz})$.  

Let the instances (corresponding to the nodes of the subtree) be $(G,k)$, $(G_1,k_1)$, $(G_2,k_2)$ and $(G_2^\prime,k_2^\prime)$. That is, $G_1={\cal R}({\cal R}(G\setminus \{x\})\setminus \{v_{yz}\})$, $G_2^\prime={\cal R}(G_1\setminus \{w\})$ and $G_2={\cal R}(G_1\setminus N[w])$ .
 
By~Lemma~\ref{lemma:maindrop}, we know that $\mu(G\setminus \{x\},k-1)\leq \mu(G,k)-\frac 1 2$. 
This implies that $\mu({\cal R}(G\setminus \{x\}),k^{\prime})\leq \mu(G,k)-\frac 1 2$ where $({\cal R}(G\setminus \{x\}),k^{\prime})$ is the instance obtained by applying the preprocessing rules on $G\setminus \{x\}$.

 By Lemma~\ref{lemma:maindrop}, we also know that including $v_{yz}$ into the vertex cover will give a further drop of $\frac 1 2$. Hence, $\mu({\cal R}(G\setminus \{x\})\setminus \{v_{yz}\},k^{\prime}-1)\leq \mu(G,k)-1$. Applying further preprocessing will not increase the measure. Hence $\mu(G_1,k_1)\leq \mu(G,k)-1$.

Now, when we branch on the vertex $w$ in the next step, we know that we use one of the rules {\bf B}\ref{bran1}, {\bf B}\ref{bran2}, {\bf B}\ref{bran3} or {\bf B}\ref{bran4}. Hence, $\mu(G_2,k_2)\leq \mu(G_1,k_1)-{\frac 3 2}$ and $\mu(G_2^\prime,k_2^\prime)\leq \mu(G_1,k_1)-{\frac 1 2}$ (since {\bf B}\ref{bran4} gives the worst branching vector). But this implies that $\mu(G_2,k_2)\leq \mu(G,k)-{\frac 5 2}$ and $\mu(G_2^\prime,k_2^\prime)\leq \mu(G,k)-{\frac 3 2}$.

This completes the analysis of the branch of rule {\bf B}\ref{bran6} where $v$ is not included in the vertex cover.

\noindent
\item
Consider the subtree in which $v$ is included in the vertex cover, by Lemma~\ref{lem:leftside} we have that 
${\cal R}(G\setminus \{v\})$ has exactly  three degree $4$ vertices, say $w_1, w_2, w_3$ and furthermore  
for any $w_i$, $i\in \{1,2,3\}$, ${\cal R}({\cal R}(G\setminus \{v\})\setminus \{w_i\})$ contains 2 degree $4$ vertices.
Since $G$ is irreducible,
we have that for any vertex $v$ in $G$, the branching rules
 {\bf B}\ref{bran1},  {\bf B}\ref{bran2}  and {\bf B}\ref{bran3}  do not apply 
on the graph ${\cal R}(G\setminus \{v\})$. Thus, we know that in the branch where we include $v$ in the vertex cover, 
the first branching rule that applies on the graph ${\cal R}(G\setminus \{v\})$ is {\bf B}\ref{bran4}.  Without loss of generality, we assume that  
{\bf B}\ref{bran4}  is applied on the vertex $w_1$.  Thus, in the branch where we include $w_1$ in the vertex 
cover, we know that ${\cal R}({\cal R}(G\setminus \{v\})\setminus \{w_1\})$ contains $w_2$ and $w_3$ as degree $4$ vertices,
This implies that in the graph ${\cal R}({\cal R}(G\setminus \{v\})\setminus \{w_1\})$ one of the branching rules 
{\bf B}\ref{bran1},  {\bf B}\ref{bran2}, {\bf B}\ref{bran3} or {\bf B}\ref{bran4} apply on a vertex $w_1^*$.  Hence, we combine the execution of the rule {\bf B}\ref{bran6} along with the subsequent executions of {\bf B}\ref{bran4} and one of the rules 
{\bf B}\ref{bran1}, {\bf B}\ref{bran2}, {\bf B}\ref{bran3} or {\bf B}\ref{bran4} (see Fig.~\ref{fig:B5B6}).


We let the instances corresponding to the nodes of this subtree be $(G,k)$, $(G_1,k_1)$, $(G_2,k_2)$, $(G_2^\prime,k_2^\prime)$, $(G_3,k_3)$ and $(G_3^\prime,k_3^\prime)$, where $G_1={\cal R}(G\setminus \{v\})$, $G_2={\cal R}(G_1\setminus N[w_1])$, $G_2^\prime={\cal R}(G_1\setminus \{w_1\})$,  $G_3={\cal R}(G_2^\prime\setminus N[w_1^*])$ and  $G_3^\prime={\cal R}(G_2^\prime \setminus \{w_1^*\})$.

Lemma~\ref{lemma:maindrop}, and the fact that preprocessing rules do not increase the measure implies that $\mu(G_1,k_1)\leq \mu(G,k)$.

Now, since {\bf B}\ref{bran4} has been applied to branch on $w_1$, the analysis of the drop of measure due to {\bf B}\ref{bran4} shows that $\mu(G_2,k_2)\leq \mu(G_1,k_1)-{\frac 3 2}$ and $\mu(G_2,k_2)\leq \mu(G_1,k_1)-{\frac 1 2}$. Similarly, since, in the graph $G_2^\prime$, we branch on vertex $w_1^*$ using one of the rules {\bf B}\ref{bran1},  {\bf B}\ref{bran2}, {\bf B}\ref{bran3} or {\bf B}\ref{bran4}, we have that $\mu(G_3,k_3)\leq \mu(G_2^\prime,k_2^\prime)-{\frac 3 2}$ and $\mu(G_3^\prime,k_3^\prime)\leq \mu(G_2^\prime,k_2^\prime)-{\frac 1 2}$. 

Combining these, we get that $\mu(G_3,k_3)\leq \mu(G,k)-{\frac 5 2}$ and $\mu(G_3^\prime,k_3^\prime)\leq \mu(G,k)-{\frac 3 2}$.  This completes the analysis of  rule {\bf B}\ref{bran6} where $v$ is included in the vertex cover. Combining the analysis for both the cases results in a branching vector of $({\frac 3 2 }, {\frac 5 2}, {\frac 5 2 }, {\frac 3 2}, 2)$ for the rule {\bf B}\ref{bran6}.

\end{itemize}

\noindent
Finally, we combine all the above results to obtain the following theorem. 

\begin{theorem}\label{thm:maintheoremsecond}
{\sc Vertex Cover above LP} can be solved in time $O^*({(2.3146)}^{k-vc^*(G)})$.
\end{theorem}
\begin{proof}
Let us fix $\mu=\mu(G,k)=k-vc^*(G)$.  We have thus shown that the preprocessing rules do not increase the measure. Branching rules 
{\bf B}\ref{bran1} or {\bf B}\ref{bran2} or {\bf B}\ref{bran3} results in a $(1, 1)$ decrease in $\mu(G,k)=\mu$, resulting in the recurrence 
$T(\mu) \leq T(\mu-1) + T(\mu-1)$  which solves to 
$2^{\mu} = 2^{k-vc^*(G)}$. 

Branching rule {\bf B}\ref{bran4} results in a $(\frac{1}{2},\frac{3}{2})$ decrease in $\mu(G,k)=\mu$, resulting in the recurrence 
$T(\mu) \leq T(\mu- \frac{1}{2}) + T(\mu-\frac{3}{2})$  which solves to 
$2.1479^{\mu} = 2.1479^{k-vc^*(G)}$.

 Branching rule {\bf B}\ref{bran5} combined with the next step in the algorithm results in a $(1,\frac{3}{2},\frac{3}{2})$ 
 branching vector,
resulting in the recurrence $T(\mu) \leq T(\mu- 1) + 2 T(\mu-\frac{3}{2})$  which solves to 
$2.3146^{\mu} = 2.3146^{k-vc^*(G)}$. 

We analyzed the way algorithm works after an application of branching rule {\bf B}\ref{bran6} before Theorem~\ref{thm:maintheoremsecond}. 
An overview of drop in measure is given in Figure~\ref{fig:droptable}. 


This leads to a $(\frac{3}{2},\frac{5}{2},2,\frac{3}{2},\frac{5}{2})$ 
branching vector,
resulting in the recurrence $T(\mu) \leq T(\mu- 1) + 2 T(\mu-\frac{3}{2})$  which solves to 
$2.3146^{\mu} = 2.3146^{k-vc^*(G)}$. 

Thus, we get an 
\runtime{(k-vc^*(G)} algorithm for {\sc Vertex Cover above LP}. 
\end{proof}

\section{Applications}
In this section we give several applications of the algorithm developed for {\sc Vertex Cover above LP}. 

\subsection{An algorithm for \agvcfull\ }
Since the value of the LP relaxation is at least the size of the maximum matching, our algorithm also runs in time $O^*(2.3146^{k-m})$ where $k$ is the size of the minimum vertex cover and $m$ is the size of the maximum matching.

\begin{theorem}\label{thm:agvc}
{\sc Above Guarantee Vertex Cover} can be solved in time $O^*(2.3146^{\ell})$ time, where $\ell$ is the excess of the minimum vertex cover size above the size of the maximum matching.
\end{theorem}
\noindent
Now by the known reductions in~\cite{GottlobS08,MarxR09,RRS11} (see also Figure~\ref{fig:zoo}) we get the following corollary to Theorem~\ref{thm:agvc}. 
\begin{corollary}
\label{cor:manyproblems}
 {\sc Almost $2$-SAT},  {\sc Almost $2$-SAT($v$)}, {\sc RHorn-Backdoor Detection Set} can be solved in time $O^*(2.3146^{k})$, and 
 {\sc KVD$_{pm}$} can be solved in time $O^*(2.3146^{\frac{k}{2}})=O^*(1.5214^k)$.
\end{corollary}

\subsection{Algorithms for {\sc Odd Cycle Transversal} and {\sc Split Vertex Deletion}}
We describe a generic algorithm for both {\sc Odd Cycle Transversal} and {\sc Split Vertex Deletion}. 
Let $X,Y\in \{\mbox{Clique, Independent Set}\}$. A graph $G$ is called an {\em $(X,Y)$-graph} if its vertices can be partitioned 
into $X$ and $Y$. Observe that when $X$ and $Y$ are both \emph{independent set}, this corresponds to a {\em bipartite graph} 
and when $X$ is \emph{clique} and $Y$ is \emph{independent set}, this corresponds to a {\em split graph}. In this section we outline an algorithm that runs in time $O^*(2.3146^k)$ and solves the following problem.

\begin{center}
\begin{boxedminipage}{.9\textwidth}

\textsc{(X,Y)-Transversal Set}\vspace{10 pt}

\begin{tabular}{ r l }

\textit{~~~~Instance:} &  An undirected graph $G$ and a positive integer $k$.\\
\textit{Parameter:} & $k$.\\
\textit{Problem:} & Does $G$ have a vertex subset $S$ of size at most $k$ such that \\ 
& its deletion leaves a $(X,Y)$-graph?
 \\
\end{tabular}
\end{boxedminipage}
\end{center}


We solve the {\sc (X,Y)-Transversal Set} problem by using a reduction to \agvc\ that takes $k$ to $k$ \cite{saketthesis08}.  
We give the reduction here for the sake of the completeness.  Let $X,Y\in \{\mbox{Clique, Independent Set}\}$

\begin{quote}
{\bf Construction :} Given a graph $G=(V,E)$ and $(X,Y)$, we construct a graph $H(X,Y)=(V(X,Y),E(X,Y))$ as follows. We take two
 copies of $V$ as the vertex set of $H(X,Y)$, that is, $V(X,Y)=V_1\cup V_2$ where $V_i=\{u_i ~|~u \in V\}$ for $1\leq i \leq
  2$. The set $V_1$ corresponds to $X$ and the set $V_2$ corresponds to $Y$. The edge set of $H(X,Y)$ depends on $X$ or $Y$ being \emph{clique} or  \emph{independent set}. 
 If $X$ is \emph{independent set}, then the graph induced on $V_1$ is made isomorphic to $G$, that is, 
for every edge $(u,v)\in E$ we include the edge $(u_1,v_1)$  in $E(X,Y)$.  If $X$ is \emph{clique}, then the graph induced on $V_1$ is isomorphic to the complement of 
$G$, that is,  for every non-edge $(u,v)\notin E$ we include an edge $(u_1,v_1)$ in  $E(X,Y)$. Hence, if $X$(respectively $Y$) is \emph{independent set}, then we make the corresponding $H(X,Y)[V_i]$ isomorphic to the graph $G$ and otherwise, we make $H(X,Y)[V_i]$ isomorphic to the complement of $G$. 
Finally, we add the perfect matching $P=\{(u_1,u_2)~|~u \in V\}$ to $E(X,Y)$. This completes the construction of $H(X,Y)$. 
\end{quote}

We first prove the following lemma, which relates the existence of an $(X,Y)$-induced  
subgraph in the graph $G$ to the existence of an independent set in the associated auxiliary 
graph $H(X,Y)$.  We use this lemma to relate  {\sc (X,Y)-Transversal Set} 
to \agvc. 
\begin{lemma}
\label{maxkcol}
Let $X,Y\in \{\mbox{Clique, Independent Set}\}$ and $G=(V,E)$ be a graph on $n$ vertices.
Then, $G$ has an $(X,Y)$-induced subgraph of size $t$ 
if and only if $H(X,Y)$ has an independent set of size $t$.
\end{lemma}
\begin{proof}
Let $S \subseteq V$ be such that $G[S]$ is an $(X,Y)$-induced 
subgraph of size $t$. Let $S$ be partitioned as 
$S_1$ and $S_2$ such that $S_1$ is $X$ and $S_2$ is $Y$.  
We also know that $H(X,Y)=(V(X,Y),E(X,Y))$. 
Consider the image of $S_1$ in $V_1$ and $S_2$ in $V_2$. Let the images be $S_1^H$ and $S_2^H$ respectively. 
%
We claim that $S_1^H\cup S_2^H$  
is an independent set of size $t$ in the graph $H(X,Y)$. To see that this is indeed the case, it is enough to observe that $S_i$'s partition $S$, $H[V_i]$ is a 
copy of $G$ or a copy of the complement of $G$ based on the nature of 
$X$ and $Y$.  Furthermore, the only edges between any pair of 
copies of $G$ or $\overline{G}$ in $H(X,Y)$, are of the form $(u_1,u_2), u \in V$, that is, 
the matching edges.

Conversely, let $K$ be an independent set in $H(X,Y)$ of size $t$ and let $K$ be decomposed 
as $K_i$, $1 \leq i \leq 2$, where $K_i=K \cap V_i$.  
Let $B_i$ be the set of vertices of $G$ which correspond to the vertices of $K_i$, that is, $B _i=\{u ~|~u \in V,~u_i\in K_i \}$ for $1 \leq i \leq 2$.  
Observe that, for any $u \in V$, $K$ contains at most one of the two copies of $u$, that is, 
$|K\cap \{u_1,u_2 \}|\leq 1$ for any $u\in V$. Hence, $\vert B_1\vert+\vert B_2\vert=t$.
Recall that, if $X$ is \emph{independent set}, then $H(X,Y)[V_1]$ is a copy of 
$G$ and hence $B_1$ is an independent set in $G$ and if $X$ is \emph{clique}, 
then $H(X,Y)[V_1]$ is a copy of the complement of $G$ and hence $K_1$ is an independent set in $\overline{G}$, 
and thus $B_1$ induces a clique in $G$. The same can be argued for the two cases for $Y$. 
Hence, the graph $G[B_1\cup B_2]$ is indeed an $(X,Y)$-graph of size $t$. This completes the proof of the lemma.
\end{proof}
\noindent
Using Lemma~\ref{maxkcol}, we obtain the following result. 
\begin{lemma}
Let $X,Y\in \{\mbox{Clique, Independent Set}\}$ and $G$ be a graph on $n$ vertices.
Then $G$ has a set of vertices of size at most $k$  
whose deletion leaves an $(X,Y)$-graph 
if and only if  $H(X,Y)$ has a vertex cover of size at most $n+k$, where $n$ is the size of the perfect 
matching of $H(X,Y)$.
\end{lemma}
\begin{proof}
By Lemma~\ref{maxkcol} we have that $G$ has an $(X,Y)$-induced subgraph of size $t$ 
if and only if $H(X,Y)$ has an independent set of size $t$. Thus, $G$ has an  $(X,Y)$-induced subgraph of size $n-k$ 
if and only if $H(X,Y)$ has an independent set of size $n-k$. But this can happen if and only if $H(X,Y)$ has a vertex cover of size at most $2n-(n-k)=n+k$.  This proves the claim. 
\end{proof}

\noindent
Combining the above lemma with Theorem~\ref{thm:agvc}, we have the following. 
\begin{theorem}
{\sc (X,Y)-Transversal Set}  can be solved in time $O^*(2.3146^k)$. 
\end{theorem}
\noindent
As a corollary to the above theorem we get the following new results. 
\begin{corollary}
{\sc Odd Cycle Transversal} and {\sc Split Vertex Deletion} can be solved in time $O^*(2.3146^k)$. 
\end{corollary}
\noindent
Observe that the reduction from {\sc Edge Bipartization} to {\sc Odd Cycle Transversal} represented in Figure~\ref{fig:zoo}, along with the above corollary implies that {\sc Edge Bipartization} can also be solved in time $O^*(2.3146^k)$. However, we note that Guo et al.~\cite{Guo:2006} have given an algorithm for this problem running in time $O^*(2^k)$.

\subsection{An algorithm for {\sc \Konig\ Vertex Deletion}}
A graph $G$ is called \Konig\ if the size of a minimum vertex cover 
equals that of a maximum matching in the graph.   Clearly bipartite graphs are \Konig\, but there are non-bipartite graphs that are \Konig\ (a triangle with an edge attached to one of its vertices, for example). Thus the {\sc \Konig\ Vertex Deletion} problem, as stated below, is closely connected to {\sc Odd Cycle Transversal}.

\begin{center}
\begin{boxedminipage}{.9\textwidth}
\textsc{\Konig\ Vertex Deletion (KVD)}\vspace{10 pt}

\begin{tabular}{ r l }

\textit{~~~~Instance:} &  An undirected graph $G$ and a positive integer $k$.\\
\textit{Parameter:} & $k$.\\
\textit{Problem:} & Does $G$ have a vertex subset $S$ of size at most $k$ such  \\
&  that $G\setminus S$ is a \Konig\ graph?
 \\
\end{tabular}
\end{boxedminipage}
\end{center}
If the input graph $G$ to {\sc \Konig\ Vertex Deletion}  has a perfect matching then this problem is called 
{\sc KVD$_{pm}$}. By Corollary~\ref{cor:manyproblems}, we already know that {\sc KVD$_{pm}$} has an algorithm with running time 
$O^*(1.5214^k)$ by a polynomial time reduction to \agvc, that takes $k$ to $k/2$. However, there is no known reduction 
if we do not assume that the input graph has a perfect matching and it required several interesting structural theorems in~\cite{MishraRSSS10} to show that 
{\sc KVD} can be solved as fast as \agvc. Here, we outline an algorithm for {\sc KVD} that runs in $O^*(1.5214^k)$ and uses an 
interesting reduction rule. However, for our algorithm we take a detour and solve a slightly different, although equally interesting problem. Given a graph, a  set $S$ of vertices is called {\em {\Konig} vertex deletion set (kvd set)} 
if its removal leaves a {\Konig} graph. The auxiliary problem we study is following.

\begin{center}
\begin{boxedminipage}{.9\textwidth}
\textsc{Vertex Cover Param by KVD}\vspace{10 pt}

\begin{tabular}{ r l }

\textit{~~~~Instance:} &  An undirected graph $G$, a  {\Konig} vertex deletion set $S$ of size \\ 
& at most $k$   and a positive integer $\ell$.\\
\textit{Parameter:} & $k$.\\
\textit{Problem:} & Does $G$ have a vertex cover of size at most $\ell$?
 \\
\end{tabular}
\end{boxedminipage}
\end{center}

This fits into the recent study of problems parameterized by other structural parameters. See, for example
{\sc Odd Cycle Transversal} parameterized by various structural parameters~\cite{abs-1107-3658} or {\sc  Treewidth} parameterized by 
vertex cover~\cite{BodlaenderJK11} or {\sc Vertex Cover} parameterized by feedback vertex set~\cite{JansenB11}. 
\noindent
For our proofs we will use the following characterization 
of \Konig\ graphs.

\begin{lemma}{\rm \cite[Lemma~$1$]{MishraRSSS10}}
 \label{lem:konig_charac}
A graph~$G=(V,E)$ is \Konig\ if and only if there exists a bipartition
of~$V$ into~$V_1 \uplus V_2$, with~$V_1$ a vertex cover of~$G$ such that
there exists a matching across the cut~$(V_1,V_2)$ saturating every
vertex of~$V_1$.
\end{lemma}

Note that in {\sc Vertex Cover param by KVD}, $G\setminus S$ is a \Konig\ graph. So one could branch on all subsets of $S$ to include in the output vertex cover, and for those elements not picked in $S$, we could pick its neighbors in $G\setminus S$ and delete them. However, the resulting graph need not be \Konig\ adding to the complications. Note, however, that such an algorithm would yield an $O^*(2^k)$ algorithm for {\sc Vertex Cover Param by OCT}. 
That is,  if $S$ were an odd cycle transversal then the resulting graph after deleting the neighbors of vertices not picked from $S$ will remain a bipartite graph, where an optimum vertex cover can be found in polynomial time. 

Given a graph~$G=(V,E)$ and two disjoint vertex subsets~$V_1, V_2$
of~$V$, we let~$(V_1,V_2)$ denote the bipartite graph with vertex 
set~$V_1 \cup V_2$ and edge set~$\{\{u,v\}: \mbox{$\{u,v\} \in E$ 
and~$u \in V_1, v \in V_2$} \}$. 
Now, we describe an algorithm based on Theorem~\ref{thm:main theorem}, that solves {\sc Vertex Cover param by KVD} in time $O^*(1.5214^{k})$. 
\begin{theorem}\label{thm:vc par by kvd}
{\sc Vertex Cover Param by KVD} can be solved in time $O^*({1.5214}^{k})$.
\end{theorem} 
\begin{proof}
Let $G$ be the input graph, $S$ be a kvd set of size at most $k$. We first apply Lemma~\ref{lem:compute good set} on $G=(V,E)$ and obtain 
an optimum solution to  LPVC($G$) such that all $\frac{1}{2}$ is the unique optimum solution to LPVC($G[V^x_{1/2}]$). Due to Lemma~\ref{lem:classicalNT}, this implies that there exists a minimum vertex cover of $G$ that contains all the vertices in $V^x_1$ and none of the vertices in $V^x_0$. Hence, the problem reduces to finding a vertex cover of size $\ell^\prime=\ell-|V^x_1|$ for the graph $G^\prime=G[V^x_{1/2}]$. Before we describe the rest of the algorithm, we prove the following lemma regarding kvd sets in $G$ and $G^\prime$ which shows 
that if $G$ has a kvd set of size at most $k$ then so does $G^\prime$. Even though this looks straight forward, the fact that {\Konig} graphs are not hereditary (i.e. induced subgraphs of {\Konig} graphs need not be {\Konig}) makes this a non-trivial claim to prove.

\begin{lemma}\label{lem:kvd} Let $G$ and $G^\prime$ be defined as above. Let $S$ be a kvd set of graph $G$ of size at most $k$. Then, there is a 
kvd set of graph $G^\prime$ of size at most $k$.
\end{lemma}
\begin{proof} It is known that the sets $(V^x_0,V^x_1,V^x_{1/2})$ form a \emph{crown decomposition} of the graph $G$~\cite{ChlebikC08}. In other words, $N(V^x_0)=V^x_1$ and there is a matching saturating $V^x_1$ in the bipartite graph $(V^x_1,V^x_0)$. The set $V^x_0$ is called the \emph{crown} and the set $V^x_1$ is called the \emph{head} of the decomposition. For ease of presentation, we will refer to the set $V^x_0$ as $C$, $V^x_1$ as $H$ and the set $V^x_{1/2}$ as $R$. In accordance with Lemma~\ref{lem:konig_charac}, let $A$ be the minimum vertex cover and let $I$ be the corresponding independent set of $G\setminus S$ such that there is a matching saturating $A$ across the bipartite graph $(A,I)$. First of all, note that if the set $S$ is disjoint from $C\cup H$, $H\subseteq A$, and $C\subseteq I$, we are done, since the set $S$ itself can be taken as 
a kvd set for $G^\prime$.  This last assertion follows because there exists a matching saturating $H$ into $C$. Hence, we may assume that this is not the case. However, we will argue that given a kvd set of $G$ of size at most $k$ we will always be able to modify it in a way that it is of size at most $k$, it is disjoint from $C\cup H$, $H\subseteq A$, and $C\subseteq I$. This will allow us to prove our lemma. Towards this, 
we now consider the set $H^\prime=H\cap I$ and consider the following two cases.

\begin{enumerate}
\item $H^\prime$ is empty. We now consider the set $S^\prime=S\setminus (C\cup H)$ and claim that $S^\prime$ is also a  kvd set of $G$ of size at most $k$ such that $G\setminus S^\prime$ has a vertex cover $A^\prime=(A\setminus C)\cup H$ with the corresponding independent set being $I^\prime=I\cup C$. In other words, we move all the vertices of $H$ to $A$ and the vertices of $C$ to $I$. Clearly, the size of the set $S^\prime$ is at most that of $S$. The set $I^\prime$ is independent since $I$ was intially independent, and the newly added vertices have edges only to vertices of $H$, which are not in $I^\prime$. Hence, the set $A^\prime$ is indeed a vertex cover of $G\setminus S^\prime$. Now, the vertices of $R$, which lie in $A$, (and hence $A^\prime$) were saturated by vertices not in $H$, since $H\cap I$ was empty. Hence, we may retain the matching edges saturating these vertices, and as for the vertices of $H$, we may use the matching edges given by the crown decomposition to saturate these vertices and thus there is a matching saturating every vertex in $A^\prime$ across the bipartite graph $(A^\prime,I^\prime)$. Hence, we now have a kvd set $S^\prime$ disjoint from $C\cup H,$ such that $H$ is part of the vertex cover and $C$ lies in the independent set of the {\Konig} graph $G\setminus S^\prime$.

\begin{figure}[top]
\begin{center}
\includegraphics[height=150 pt, width=220 pt]{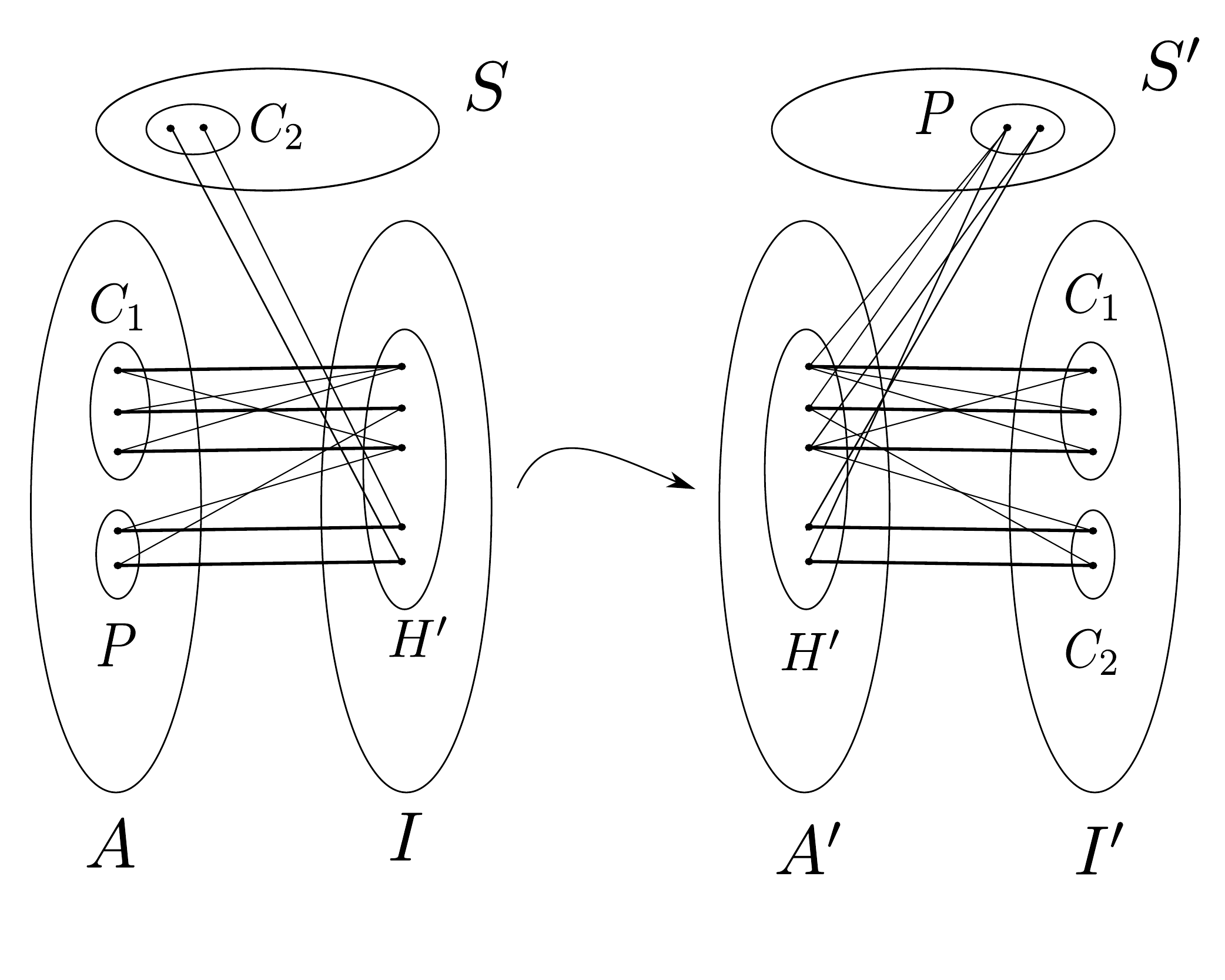}
\caption{An illustration of case 2 of Lemma~\ref{lem:kvd}}
\label{fig:lemma8}
\end{center}
\end{figure}

\item $H^\prime$ is non empty. Let $C_1$ be the set of vertices in $A\cap C$ which are adjacent to $H^\prime$ (see Fig.~\ref{lem:kvd}) , let $C_2$ be the set of vertices in $C\cap S$, which are adjacent to $H^\prime$, and let $P$ be the set of vertices of $R\cap A$ which are saturated by vertices of $H^\prime$ in the bipartite graph $(A,I)$. We now consider the set $S^\prime=(S\setminus C_2)\cup P$ and
claim that $S^\prime$ is also a kvd set of $G$ of size at most $k$ such that $G\setminus S^\prime$ has a minimum vertex cover $A^\prime=(A\setminus (C_1\cup P))\cup H^\prime$ with the corresponding independent set being $I^\prime=(I\setminus H^\prime)\cup (C_1\cup C_2)$. In other words, we move the set $H^\prime$ to $A$, the sets $C_1$ and $C_2$ to $I$ and the set $P$ to $S$. The set $I^\prime$ is independent since $I$ was independent and the vertices added to $I$ are adjacent only to vertices of $H$, which are not in $I^\prime$. Hence, $A^\prime$ is indeed a vertex cover of $G\setminus S^\prime$. To see that there is still a matching saturating $A^\prime$ into $I^\prime$, note that any vertex previously saturated by a vertex not in $H$ can still be saturated by the same vertex. As for vertices of $H^\prime$, which have been newly added to $A$, they can be saturated by the vertices in $C_1\cup C_2$. Observe that $C_1\cup C_2$  is precisely the neighborhood of $H^\prime$ in $C$ and since there is a matching saturating $H$ in the bipartite graph $(H,C)$ by Hall's Matching Theorem we have that for every subset $\hat H\subseteq H$, $\vert N(\hat H)\cap (C_1\cup C_2)\vert\geq \vert \hat H\vert$.
Hence, by Hall's Matching Theorem  there is a matching saturating $A^\prime$ in the bipartite graph $(A^\prime,I^\prime)$. It now remains to show that $\vert S^\prime\vert\leq k$.

Since $N(H^\prime)=C_1\cup C_2$ in the bipartite graph $(C, H)$, we know that $\vert C_1\vert+\vert C_2\vert\geq \vert H^\prime\vert$. In addition, the vertices of $C_1$ have to be saturated in the bipartite graph $(A,I)$ by vertices in $H^\prime$. Hence, we also have that $\vert C_1\vert+\vert P\vert\leq \vert H^\prime\vert$. This implies that $\vert C_2\vert\geq \vert P\vert$. Hence, $\vert S^\prime\vert\leq \vert S\vert\leq k$. This completes the proof of the claim. But now, notice that we have a kvd set of size at most $k$ such that there are no vertices of $H$ in the independent set side of the corresponding {\Konig} graph. Thus, we have fallen into Case 1, which has been handled above.
\end{enumerate}
This completes the proof of the lemma.
\end{proof}
\noindent
We now show that $\mu=vc(G^\prime)-vc^*(G^\prime)\leq \frac{k}{2}$.
Let $O$ be a kvd set of $G^\prime$ and define $G^{\prime\prime}$ as the K\'{o}nig graph $G^\prime\setminus O$. It is well known that
in {\Konig} graphs, $|M|=vc(G^{\prime\prime})=vc^*(G^{\prime\prime})$, where $M$ is a maximum matching in the graph
$G^{\prime\prime}$.
This implies that $vc(G^\prime)\leq vc(G^{\prime\prime})+|O|=\vert M\vert+\vert O\vert$.
But, we also know that 
$vc^*(G^\prime)\geq \vert M\vert+\frac{1}{2}(\vert O\vert)$ and hence,
 $vc(G^\prime)-vc^*(G^\prime)\leq \frac{1}{2}(\vert O\vert)$. 
 By Lemma \ref{lem:kvd}, we know that there is an $O$ such that $\vert O\vert\leq k$ and hence,
$vc(G^\prime)-vc^*(G^\prime)\leq \frac{k}{2}$.
 
By Corollary~\ref{cor:compute min vc}, we can find a minimum vertex cover of $G^\prime$ in time 
$O^*(2.3146^{vc(G^\prime)-vc^*(G^\prime)})$ and hence in time $O^*(2.3146^{k/2})$. If the size of the minimum vertex cover obtained for $G^\prime$ is at most $\ell^\prime$, then we return yes else we return no. This completes the proof of the theorem.
\end{proof}
\noindent
It is known that, given a minimum vertex cover, a minimum sized kvd set can be computed in polynomial time~\cite{MishraRSSS10}. Hence, Theorem~\ref{thm:vc par by kvd} has the following corollary.

\begin{corollary}
{\sc KVD} can be solved  in time $O^*(1.5214^{k})$.
\end{corollary}
\noindent
Since the size of a minimum Odd Cycle Transversal is at least the size of a minimum Konig Vertex Deletion set, we also have 
the following corollary. 

\begin{corollary}
{\sc Vertex Cover Param by OCT} can be solved in time $O^*({1.5214}^{k})$.
\end{corollary}


\subsection{A simple improved kernel for {\sc Vertex Cover}}
We give a kernelization for {\sc Vertex Cover} based on Theorem~\ref{thm:main theorem} as follows. Exhaustively, apply the 
Preprocessing rules~\ref{red:NT_reduction} through ~\ref{red:struction} (see Section~\ref{sec:mainalgo}). When the rules no longer apply, if $k-vc^*(G)\leq \log k$, then solve the problem in time $O^*(2.3146^{\log k})=O(n^{O(1)})$. Otherwise, just return the instance. We claim that the number of vertices in the returned instance is at most $2k-2\log k$. Since $k-vc^*(G)>\log k$, $vc^*(G)$ is upper bounded by $k-\log k$. But, we also know that when Preprocessing Rule~\ref{red:NT_reduction} is no longer applicable, all $\frac{1}{2}$ is the unique optimum to LPVC($G$) and hence, the number of vertices in the graph $G$ is twice the value of the optimum value of LPVC($G$). Hence, $\vert V\vert= 2vc^*(G)\leq 2(k-\log k)$. Observe that by the same method we can also show that in the reduced instance the number of vertices is upper bounded by $ 2k-c \log k$ for any fixed constant $c$. Independently, Lampis~\cite{Lampis11} has also shown an upper bound of $2k- c\log k$ on the size of a kernel for {\sc vertex cover} for any fixed constant $c$. 

\section{Conclusion}
We have demonstrated that using the drop in LP values to analyze in branching algorithms can give powerful results for parameterized complexity. 
We believe that our algorithm is the beginning of a race to improve the running time bound for \agvc\ and possibly for the classical {\sc vertex cover} problem, for which there has been no progress in the last several years after an initial plethora of results.

Our other contribution is to exhibit several parameterized problems that are equivalent to or reduce to \agvc\ through parameterized reductions. We observe that as the parameter change in these reductions are linear, any upper or lower bound results for kernels for one problem will carry over for the other problems too (subject to the directions of the reductions). For instance, recently, Kratsch and Wahlstr\"{o}m~\cite{KratschW11} studied the kernelization complexity of \agvc\ and obtained a randomized polynomial sized kernel for it through matroid based techniques.  This implies a randomized polynomial kernel for all the problems in this paper.

\bibliographystyle{siam}
\bibliography{references}

\newpage 
\section{Appendix: Problem Definitions}
\begin{center}
\begin{boxedminipage}{.9\textwidth}

\textsc{Vertex Cover}\vspace{10 pt}

\begin{tabular}{ r l }

\textit{~~~~Instance:} &  An undirected graph $G$ and a positive integer $k$. \\

\textit{Parameter:} & $k$.\\
\textit{Problem:} & Does $G$ have a vertex cover of of size at most $k$? 
 \\
\end{tabular}
\end{boxedminipage}
\end{center}

\begin{center}
\begin{boxedminipage}{.9\textwidth}

\textsc{\agvcfull\  (\agvc)}\vspace{10 pt}

\begin{tabular}{ r l }

\textit{~~~~Instance:} &  An undirected graph $G$, a maximum matching $M$ and \\ 
& a positive integer $\ell$. \\

\textit{Parameter:} & $\ell$.\\
\textit{Problem:} & Does $G$ have a vertex cover of of size at most $|M|+\ell$? 
 \\
\end{tabular}
\end{boxedminipage}
\end{center}

\begin{center}
\begin{boxedminipage}{.9\textwidth}

\textsc{Vertex Cover above LP}\vspace{10 pt}

\begin{tabular}{ r l }

\textit{~~~~Instance:} &  An undirected graph $G$, positive integers $k$ and $\lceil vc^*(G) \rceil $,  \\ 
& where $vc^*(G)$ is  the minimum value  of {\sc LPVC}. \\
\textit{Parameter:} & $k-\lceil vc^*(G) \rceil$.\\
\textit{Problem:} & Does $G$ have a vertex cover of of size at most $k$? 
 \\
\end{tabular}
\end{boxedminipage}
\end{center}
A graph $G$ is called an {\em bipartite} if its vertices can be partitioned 
into $X$ and $Y$ such that $X$ and $Y$ are independent sets.

\begin{center}
\begin{boxedminipage}{.9\textwidth}

\textsc{Odd Cycle Transveral (OCT)}\vspace{10 pt}

\begin{tabular}{ r l }

\textit{~~~~Instance:} &  An undirected graph $G$ and a positive integer $k$.\\
\textit{Parameter:} & $k$.\\
\textit{Problem:} & Does $G$ have a vertex subset $S$ of size at most $k$ such  \\
&  that $G\setminus S$ is a bipartite graph?
 \\
\end{tabular}
\end{boxedminipage}
\end{center}

\begin{center}
\begin{boxedminipage}{.9\textwidth}
 \textsc{Edge Bipartization (EB)}\vspace{10 pt}

\begin{tabular}{ r l }

\textit{~~~~Instance:} &  An undirected graph $G$ and a positive integer $k$.\\
\textit{Parameter:} & $k$.\\
\textit{Problem:} & Does $G$ have an edge subset $S$ of size at most $k$ such  \\
&  that $G'=(V,E\setminus S)$ is a bipartite graph?
 \\
\end{tabular}
\end{boxedminipage}
\end{center}
 
 \noindent
A graph $G$ is called an {\em split} if its vertices can be partitioned 
into $X$ and $Y$ such that $X$ is a clique and $Y$ is an independent set.

\begin{center}
\begin{boxedminipage}{.9\textwidth}

\textsc{Split Vertex Deletion}\vspace{10 pt}

\begin{tabular}{ r l }

\textit{~~~~Instance:} &  An undirected graph $G$ and a positive integer $k$.\\
\textit{Parameter:} & $k$.\\
\textit{Problem:} & Does $G$ have a vertex subset $S$ of size at most $k$ such  \\
&  that $G\setminus S$ is a split graph?
 \\
\end{tabular}
\end{boxedminipage}
\end{center}
\noindent
A graph $G$ is called an {\em \Konig} if the size of a maximum matching is equal to the 
size of a minimum vertex cover. 
 
\begin{center}
\begin{boxedminipage}{.9\textwidth}
\textsc{\Konig\ Vertex Deletion (KVD)}\vspace{10 pt}

\begin{tabular}{ r l }

\textit{~~~~Instance:} &  An undirected graph $G$ and a positive integer $k$.\\
\textit{Parameter:} & $k$.\\
\textit{Problem:} & Does $G$ have a vertex subset $S$ of size at most $k$ such  \\
&  that $G\setminus S$ is a \Konig\ graph?
 \\
\end{tabular}
\end{boxedminipage}
\end{center}

If the input graph to {\sc KVD} has a perfect matching then we call it {\sc KVD$_{pm}$}. 
 
\noindent
Given a $2$-SAT formula $\phi$ on variables $x_1,\dots,x_n$, and with clauses $C_1,\dots,C_m$, we define \emph{deleting} a clause from $\phi$ as removing the clause from the formula $\phi$ and deleting a variable from $\phi$ as removing all the clauses which involve that variable, from $\phi$. 
 
\begin{center}
\begin{boxedminipage}{.9\textwidth}

\textsc{Almost $2$-SAT}\vspace{10 pt}

\begin{tabular}{ r l }

\textit{~~~~Instance:} &  A $2$-SAT formula $\phi$ and a positive integer $k$.\\
\textit{Parameter:} & $k$.\\
\textit{Problem:} & Does there exist a set of at most $k$ clauses, whose deletion \\ 
& from $\phi$   makes the resulting formula satisfiable?
 \\
\end{tabular}
\end{boxedminipage}
\end{center}

\begin{center}
\begin{boxedminipage}{.9\textwidth}

\textsc{Almost $2$-SAT-VARIABLE VERSION} ({\sc Almost $2$}-{\sc SAT}($v$))\vspace{10 pt}

\begin{tabular}{ r l }

\textit{~~~~Instance:} &  A $2$-SAT formula $\phi$ and a positive integer $k$.\\
\textit{Parameter:} & $k$.\\
\textit{Problem:} & Does there exist a set of at most $k$ variables, whose deletion \\ 
& from $\phi$   makes the resulting formula satisfiable?
 \\
\end{tabular}
\end{boxedminipage}
\end{center}
\noindent
Given a graph $G$, a vertex subset $K$ of $G$ is said to be a {\Konig} vertex deletion ({\sc KVD}) set if the graph $G\setminus K$ is a {\Konig} graph.

\begin{center}
\begin{boxedminipage}{.9\textwidth}

\textsc{Vertex Cover Param By KVD}\vspace{10 pt}

\begin{tabular}{ r l }

\textit{~~~~Instance:} &  An undirected graph $G$, a positive integer $k$, and a set $K$, \\
& which is a {\sc KVD} set for $G$.\\
\textit{Parameter:} & $\vert K\vert$.\\
\textit{Problem:} & Does $G$ have a vertex cover of size at most $k$? \\

 \\
\end{tabular}
\end{boxedminipage}
\end{center}

\begin{center}
\begin{boxedminipage}{.9\textwidth}

\textsc{Vertex Cover Param By OCT}\vspace{10 pt}

\begin{tabular}{ r l }

\textit{~~~~Instance:} &  An undirected graph $G$, a positive integer $k$, and a set $K$, \\
& which is an OCT for $G$.\\
\textit{Parameter:} & $\vert K\vert$.\\
\textit{Problem:} & Does $G$ have a vertex cover of size at most $k$? \\

 \\
\end{tabular}
\end{boxedminipage}
\end{center}

\noindent
{\sc Horn} denotes the set of CNF formulas where each clause contains at most one positive literal. {\sc RHorn} denotes the class of 
renamable {\sc Horn} CNF formulas, that is, of  CNF formulas $F$ for which there exists a set $X\subset var(F)$ such that, 
replacing in the clauses of $F$ the literal $x$ by $\bar{x}$ and the literal $\bar{x}$ by $x$ whenever $x\in X$,  
yields a Horn formula. The set $var(F)$ contains the variables contained in $F$.  
Obviously, {\sc RHorn} properly contains {\sc Horn}.  For a CNF formula $F$ and a set of variables  $B\subseteq var(F)$  
let $F\setminus B$ denote
the CNF formula $\{C\setminus (B\cup \overline{B})~:~C\in F\}$, that is, set of clauses obtained after deleting the variables and its negation in the set $B$.  For a formula $F$, we say that a set $B\subseteq var(F)$ is {\em deletion {\sc RHorn}-backdoor set} if $F\setminus B$ is in {\sc RHorn}.

\begin{center}
\begin{boxedminipage}{.9\textwidth}

\textsc{RHorn-Backdoor Detection Set (RHBDS)}\vspace{10 pt}

\begin{tabular}{ r l }

\textit{~~~~Instance:} &  A CNF formula $\phi$ and a positive integer $k$.\\
\textit{Parameter:} & $k$.\\
\textit{Problem:} & Does there exists a deletion {\sc RHorn}-backdoor set of size \\ 
& at most $k$?
 \\
\end{tabular}

\end{boxedminipage}
\end{center}

%
%
%
%
%

\end{document}